\documentclass[journal]{IEEEtran}

\usepackage[english]{babel}
\usepackage[utf8x]{inputenc}
\usepackage[T1]{fontenc}
\usepackage{cite}

\usepackage{amsmath,amssymb,mathrsfs,bbm,amsthm}
\usepackage{graphics,graphicx}
\usepackage[colorinlistoftodos]{todonotes}
\usepackage[colorlinks=true, allcolors=blue]{hyperref}
\usepackage{multirow}
\usepackage[lined,linesnumbered,ruled]{algorithm2e}
\usepackage{booktabs,diagbox}
\usepackage{balance}

\usepackage{mathbbol}
\DeclareSymbolFontAlphabet{\amsmathbb}{AMSb}%

\newcommand{\cp}[1]{\ifmmode {\mathcal{#1}}\else ${\mathcal{#1}}$\fi}
	
\newcommand{\bA}{\boldsymbol{A}}

\newcommand{\bI}{\boldsymbol{I}}

\newcommand{\bM}{\boldsymbol{M}}

\newcommand{\bX}{\boldsymbol{X}}
\newcommand{\bY}{\boldsymbol{Y}}

\newcommand{\ba}{\boldsymbol{a}}

\newcommand{\bm}{\boldsymbol{m}}

\newcommand{\be}{\boldsymbol{e}}

\newcommand{\by}{\boldsymbol{y}}
\newcommand{\bs}{\boldsymbol{s}}

\newcommand{\bx}{\boldsymbol{x}}

\newcommand{\calM}{\mathcal{M}}
\newcommand{\calN}{\mathcal{N}}

\newcommand{\calU}{\mathcal{U}}

\newcommand{\bdelta}{\boldsymbol{\delta}}

\newcommand{\bmu}{\boldsymbol{\mu}}

\newcommand{\biota}{\boldsymbol{\iota}}

\newcommand{\bXi}{\boldsymbol{\Xi}}

\newcommand{\cb}[1]{\boldsymbol{#1}}

\newcommand{\tr}{\operatorname{tr}}
\newcommand{\cov}{\operatorname{cov}}

\usepackage{color}  %
\def\cred{\textcolor{red}}

\definecolor{darkgreen}{rgb}{0.0, 0.85, 0.0}

\renewcommand{\cred}{}

\def\cmark{\textcolor{red}}

\newtheorem{theorem}{Theorem}%

\title{Fast Unmixing and Change Detection in Multitemporal Hyperspectral Data}

\author{Ricardo~Augusto~Borsoi,~\IEEEmembership{Student~Member,~IEEE,} Tales Imbiriba, Jos\'e~Carlos~Moreira~Bermudez,~\IEEEmembership{Senior~Member,~IEEE}, C\'edric Richard,~\IEEEmembership{Senior~Member,~IEEE}
\thanks{This work has been supported by the National Council for Scientific and Technological Development (CNPq) under grants 304250/2017-1, 409044/2018-0, 141271/2017-5 and 204991/2018-8.}
\thanks{R.A. Borsoi is with the Department of Electrical Engineering, Federal University of Santa Catarina (DEE--UFSC), Florian\'opolis, SC, Brazil, and with the Lagrange Laboratory, Universit\'e  C\^ote  d'Azur, Nice, France. e-mail: \mbox{raborsoi@gmail.com}.}
\thanks{T. Imbiriba is with the ECE department of the Northeastern University, Boston, MA, USA. e-mail: \mbox{talesim@ece.neu.edu}.}
\thanks{J.C.M. Bermudez is with the DEE--UFSC, Florian\'opolis, SC, Brazil. e-mail: \mbox{j.bermudez@ieee.org}.}
\thanks{C. Richard is with the  Universit\'e  C\^ote  d'Azur,  Nice, France (e-mail: cedric.richard@unice.fr), Lagrange Laboratory (CNRS, OCA).}
\thanks{Manuscript received Month day, year; revised Month day, year.}}

\begin{document}
\maketitle

\begin{abstract}

Multitemporal spectral unmixing (SU) is a powerful tool to process hyperspectral image (HI) sequences due to its ability to reveal the evolution of materials over time and space 
in a scene. However, significant spectral variability is often observed between collection of images due to variations in acquisition or seasonal conditions. This characteristic has to be considered in the design of SU algorithms. Because of its good performance, the multiple endmember spectral mixture analysis algorithm (MESMA) has been recently used to perform SU in multitemporal scenarios arising in several practical applications. However, MESMA does not consider the relationship between the different HIs, and its computational complexity is extremely high for large spectral libraries.
In this work, we propose an efficient multitemporal SU method that exploits the high temporal correlation between the abundances to provide more accurate results at a lower computational complexity. We propose to solve the multitemporal SU problem by separately addressing the endmember selection and the abundance estimation problems. This leads to a simpler solution without sacrificing the accuracy of the results. We also propose a strategy to detect and address abrupt abundance variations \cred{in time}.
Theoretical results demonstrate how the proposed method compares to MESMA in terms of quality, and how effective it is in detecting abundance changes. This analysis provides valuable insight into the conditions under which the algorithm succeeds.
Simulation results show that the proposed method achieves state-of-the-art performance at a smaller computational cost.

\end{abstract}


\section{Introduction}

Hyperspectral images (HI) have become a central tool in an increasing number of applications due to their high spectral resolution, which offers important information about the materials in a scene~\cite{Bioucas-Dias-2013-ID307}. However, inherent limitations of imaging devices and large sensor-to-target distances typical of many applications, such as remote sensing, lead to HIs with low spatial resolution~\cite{shaw2003spectralImagRemote}. Hence, each pixel in an HI is usually a mixture of spectral signatures of different pure materials, also called \emph{endmembers} (EM)~\cite{Keshava:2002p5667}. Spectral unmixing (SU) aims to decompose an HI into a collection of endmembers and their corresponding fractional \emph{abundances}, thus revealing important information on the distribution of the materials in the scene~\cite{Bioucas2012}.

The Linear Mixing Model (LMM) is the simplest and most widely used to represent the interaction between light and the materials in a given scene~\cite{Keshava:2002p5667}.  The LMM represents the reflectance of each pixel as a convex combination of the endmembers in the scene. The combination coefficients can then be interpreted as the fractional abundances contributed to each pixel by the EMs.
The simplest form of LMM models each material in the whole scene by a single endmember. Although allowing fast and simple SU strategies, such a model fails to account for the important phenomenon of endmember variability, observed in most practical scenes~\cite{borsoi2020variabilityReview,Zare-2014-ID324-variabilityReview,somers2011variabilityReview}.

Spectral variability can be caused by many factors including, for instance, atmospheric, illumination and seasonal variations, and can propagate significant abundance and endmember estimation errors throughout the unmixing process~\cite{borsoi2020variabilityReview,Zare-2014-ID324-variabilityReview}. This motivated the use of more elaborated mixing models and algorithms that explicitly address endmember variability in SU~\cite{borsoi2020variabilityReview,Zare-2014-ID324-variabilityReview}. 
Several parametric and non-parametric models have been recently proposed which account for EM spectra variation within a single HI. These include Beta or mixture of Gaussian distributions~\cite{du2014spatialBetaCompositional,zhou2018variabilityGaussianMixtureModel}, additive~\cite{Thouvenin_IEEE_TSP_2016_PLMM} and multiplicative scaling factors~\cite{drumetz2016blindUnmixingELMM,imbiriba2018glmm,borsoi2019tensorInterpolationICASSP,Borsoi_multiscaleVar_2018}, \cred{combinations of an uniform scaling and an additive variability dictionary~\cite{hong2019augmentedLMMvariability},} reparameterizations using deep neural networks~\cite{borsoi2019deepGun}, and low-rank tensor representations~\cite{imbiriba2018ULTRA_V}.
Although these models were able to produce promising results, the most prominent strategy to deal with spectral variability in SU is to represent EMs as sets of spectral signatures, called spectral libraries or \emph{bundles}~\cite{somers2011variabilityReview}.
The spectral libraries are usually constructed a priori from laboratory or \textit{in situ} measurements. They contain variants of the spectral signatures of each material, to better represent various acquisition conditions or physicochemical compositions.
SU can then be formulated as a problem of selecting from the library the subset of spectral signatures that best represents the observed HI. This usually entails the use of either sparse unmixing~\cite{iordache2011sunsal,borsoi2018superpixels1_sparseU} or Multiple Endmember Spectral Mixture Analysis (MESMA)~\cite{roberts1998originalMESMA} algorithms. The MESMA algorithm is still the leading algorithm in practice due to its simplicity and interpretability, and because it provides good results if the spectral library adequately expresses the spectra contained in the scene~\cite{somers2011variabilityReview,borsoi2019EMlibManInterpVAE}. This allowed MESMA to be applied to many different environments and scenarios~\cite[p.1607]{somers2011variabilityReview}.

\subsection{\cred{Multitemporal SU and change detection}}

More recently, multitemporal SU has become a topic of great interest due to its ability to leverage temporal information in HI sequences for monitoring the evolution of the different materials and their distribution in the scene~\cite{somers2013invasiveHawaiiMultiTemporalBandWeighting,lippitt2018multidateMESMAshrublands,somers2013uncorrelatedBandSelectionInstabilityIndex,goenaga2013unmixingTimeSeriesPuertoRico}. 
Spectral variability becomes a more critical issue in multitemporal SU, when compared to EM variations within a single HI, because images are acquired at different instants~\cite{Zare-2014-ID324-variabilityReview,Borsoi_2018_Fusion}. 
\cred{Some works have proposed to extend parametric EM models developed for single-image SU to the multitemporal case~\cite{henrot2016dynamical,thouvenin2016online,sigurdsson2017sparseDistU,thouvenin2018hierarchicalBU,borsoi2020multitemporalUKalmanEM}. These works considered either uniform~\cite{henrot2016dynamical} or bandwise~\cite{borsoi2020multitemporalUKalmanEM} scaling variations of reference EMs, or additive temporal perturbations over a mean EM matrix~\cite{thouvenin2016online,sigurdsson2017sparseDistU,thouvenin2018hierarchicalBU}. Such works attempt to estimate the EMs from the data by employing batch (centralized~\cite{thouvenin2018hierarchicalBU,borsoi2020multitemporalUKalmanEM} or distributed~\cite{sigurdsson2017sparseDistU}) or online~\cite{henrot2016dynamical,thouvenin2016online} processing strategies.}
\cred{Nonetheless, the multitemporal extension} of MESMA remains the most frequently used solution in practical applications~\cite{somers2013invasiveHawaiiMultiTemporalBandWeighting,somers2013uncorrelatedBandSelectionInstabilityIndex,dudley2015multitemporalLibraryPhenologicalGradiantsMESMA,lippitt2018multidateMESMAshrublands}, \cred{and naturally addresses} both the spatial and temporal variations of EM spectra.

\cred{A closely related problem to SU consists of detecting and monitoring changes in material composition across an HI sequence~\cite{liu2019reviewCD}. While many methods are focused on detecting only abrupt changes~\cite{liu2019reviewCD,borsoi2021onlineGraphCPD}, having detailed subpixel abundance evolution from which both subtle and abrupt variations can be discerned is important for many applications~\cite{goenaga2013unmixingTimeSeriesPuertoRico,chakravortty2017CD_SU_mangrove}. This motivated the development of unmixing-based change detection strategies based on, e.g., linear~\cite{erturk2015simplesSU_CD}, sparse~\cite{erturk2015sparseSU_CD} and locally adaptive~\cite{goenaga2013unmixingTimeSeriesPuertoRico} SU, where the estimated abundance results can be analyzed to detect abrupt changes~\cite{guo2021changeDetUnmixing}.
The ability of MESMA to produce good quality abundance estimates while addressing spatio-temporal EM variability makes it an important tool for this task. It is also important to note that %
learning-based approaches have recently shown good performance in classification~\cite{hong2020graphCNN_hyperspClassif,hong2020invariantAttributeHyperspClassific} unmixing~\cite{hong2021unmixingNetAEC_selfSuperv,li2021AEC_SU_modelbased} and related tasks in hyperspectral imaging~\cite{rasti2020featureExtractionHyperspReview,zhu2017zhang2016deepLearningRemoteSensingReview}. Such approaches usually leverage the capability of deep neural networks to improve the quality of the results. However, the combination of simplicity, interpretability and robustness 
that underlies MESMA and related library-based approaches still makes them appealing, especially due to their potential for out-of-the-box solutions.}

\cred{Nevertheless, the} advantages offered by MESMA do not come without compromises. Besides depending on a spectral library, a significant drawback of MESMA lies in its high computational complexity, which increases very quickly with the size of the libraries and with the number of materials. This goes against the need for online processing of large amounts of hyperspectral data~\cite{ma2015BigDataRemoteSensing,chi2016BigDataRemoteSensing}, and motivates the search for new algorithms that are both efficient and accurate.
An alternating angle minimization (AAM) approach has been recently proposed to provide an approximate but accurate solution to the MESMA problem with a significant reduction in computational complexity for large libraries~\cite{heylen2016alternatingAngleMinimization}. 
Despite its merits, the AAM algorithm does not scale well with the number of materials in the scene, and thus may also lead to a large computational complexity in practical scenarios.
Moreover, AAM and other existing multitemporal MESMA methods do not exploit the temporal correlation between the abundance maps, nor do they account for abrupt abundance variations.

\subsection{\cred{Contributions and organization}}

In this paper, we propose a fast multitemporal SU algorithm, named FM-MESMA. \cred{Differently from previous works such as MESMA and AAM~\cite{roberts1998originalMESMA,heylen2016alternatingAngleMinimization}, the proposed method explores the abundance temporal information to deliver high-quality estimates at a lower complexity. Moreover, unlike previous change detection methods based on SU~\cite{erturk2015simplesSU_CD,goenaga2013unmixingTimeSeriesPuertoRico,guo2021changeDetUnmixing}, we integrate change detection in a specific EM selection stage of the algorithm, which allows us to account for EM variability while maintaining a small computational cost. In contrast to deep learning-based frameworks, the proposed method yields a simple, low-cost and robust solution which comes with rigorous theoretical guarantees.
The main contributions of this paper are:}

\begin{itemize}
    \item[a)] \cred{We explicitly characterize slow and abrupt abundance variations in the multitemporal mixing model. This allows us to exploit abundance temporal correlation to propose an efficient and accurate SU algorithm. Differently from previous works, this is achieved by performing SU in two separate tasks: EM selection and abundance estimation, each of which can be solved more efficiently. This significantly reduces the complexity of SU compared to MESMA or AAM, with little impact on the results.}
    
    \item[b)] \cred{We propose a methodology to detect pixels that undergo abrupt abundance changes based on the results of the EM selection task. Such pixels, which would otherwise degrade the performance of the proposed method, are then handled separately using a more sophisticated strategy. Unlike typical SU-based change detection, the changed pixels are identified by taking spectral variability into account but without solving the full SU problem, leading to a lower complexity. The resulting algorithm is robust and interpretable, and has only a single tuning parameter.}
    
    \item[c)] \cred{We derive theoretical guarantees concerning the performance of FM-MESMA, both in the presence and in the absence of changes. Specifically, we first show under which conditions the proposed method is guaranteed to recover the correct EMs from the library. Then, we derive conditions under which the FM-MESMA correctly detects abrupt abundance changes. These theoretical results provide insight into which conditions are necessary for the proposed method to reach an accurate result.}
    
    \item[d)] \cred{We provide an analysis of the computational complexity of FM-MESMA, MESMA and AAM, and show how they scale with the number of spectral bands, the size of the spectral libraries and the number of EMs.}
    
\end{itemize}

Simulation results with synthetic and real data illustrate the performance of FM-MESMA when compared to MESMA, AAM, and approaches that do not rely on spectral libraries. 

The paper is organized as follows. Section~\ref{sec:MESMA} reviews the linear mixing model, MESMA and its multitemporal extensions. Section~\ref{sec:proposedMethod} presents the multitemporal mixing model and our algorithm. Section~\ref{sec:theoremsAndAnalysis} provides theoretical guarantees for the reconstruction accuracy and robustness of FM-MESMA, and an analysis of its computational complexity. Section~\ref{sec:experimentalResults} presents simulation results and comparisons. Conclusions are presented in Section~\ref{sec:conclusions}.

\section{Spectral unmixing with MESMA} \label{sec:MESMA}

The basic idea behind the MESMA algorithm is to find the EMs and the fractional abundances that best represent each pixel with the LMM. The LMM represents each $L$-band pixel $\by_n\in\amsmathbb{R}^L$ of an HI, for $n=1,\ldots, N$, as a convex combination of the spectral signatures of $P$ endmembers:
\begin{align} \label{eq:LMM}
    &\by_n = \bM \ba_n + \be_n, %
    \,\,\, \text{s.t. }\,\cb{1}^\top\ba_n = 1 \text{ and } \ba_n \geq \cb{0} %
\end{align}
where matrix $\bM\in\amsmathbb{R}^{L\times P}$ contains the spectral signatures of the EMs as its columns $\bm_{\cred{p}}$, $\cred{p}=1,\ldots,P$, $\ba_n$ is the abundance vector, and $\be_n$ is an additive noise term.

Most SU strategies use a single spectral signature to represent each material in the scene. This can lead to significant abundance estimation errors in the presence of spectral variability.
MESMA, on the other hand, considers $P$ spectral libraries known \textit{a priori}, one for each endmember, defined as:
\begin{align} \label{eq:specLibEach}
    \calM_p = \{\bm_{p,1},\ldots,\bm_{p,C_p}\}, \,\, \bm_{p,j}\in\amsmathbb{R}^L, \,\, p=1,\ldots,P
\end{align}
where $C_p$ is the number of available variations of spectral signatures of the $p$-th endmember. Then, for each pixel $n$, only one endmember is selected from each set $\calM_p$ to compose the endmember matrix $\bM_n$ for that pixel. 

We define the set of all possible endmember matrices that can be composed this way as:
\begin{align} \label{eq:specLibFull}
	\calM {}={} \Big\{\big[\bm_1,\ldots,\bm_P\big]\,:\,
	\bm_p\in\mathcal{M}_p,\,p=1,\ldots,P\Big\}
\end{align}

Assuming $\calM$ known, the MESMA SU problem corresponds to the search for the EM matrix (also called EM model) in $\calM$ that best represents each pixel in the scene. This translates into the following optimization problem for the $n$-th pixel:
\begin{align} \label{eq:mesma}
	\min_{\bM_n\in\calM}
    &  \mathop{\min}_{\ba_{n}} \,
    \big\|\by_n - \bM_n\ba_n\big\|,
    \text{ s.t. } \,  \ba_{n} \geq\cb{0}, \, \cb{1}^\top\ba_{n} = 1.
\end{align}

Despite its widespread use and good performance in practical scenarios, the computational cost of MESMA is extremely high. Since solving \eqref{eq:mesma} amounts to perform SU for every possible matrix $\bM_n$ extracted from $\calM$, its computational complexity scales with the product of the sizes of libraries $\mathcal{M}_p$ as it consists of solving $\prod_{p=1}^P |\mathcal{M}_p|$ FCLS (Fully Constrained Least Squares) problems~\cite{heylen2016alternatingAngleMinimization}.

Several works have attempted to circumvent this limitation by seeking approximate solutions to~\eqref{eq:mesma}.
The first approaches consisted of solving~\eqref{eq:mesma} for matrices $\bM_n$ randomly chosen from $\calM$ until obtaining a reconstruction error below a threshold and well distributed across all spectral bands~\cite{roberts1998originalMESMA}.
Another approach ignores both constraints in~\eqref{eq:mesma} and performs unconstrained least squares for every possible $\bM_n \in \calM$, and then selects the EM model resulting in the smallest reconstruction error without any negative abundances~\cite{combe2008MELSUMunmixingMARS}.

Although these approaches are simple, \cred{they are not guaranteed to achieve a good accuracy at a reduced processing cost, as only a relatively small subset of $\calM$ could be tested. Also, ignoring the abundance constraints can make the results more sensitive to noise.}
Recent strategies attempt to provide low-complexity alternatives to MESMA with minimal impact on unmixing results.
For instance, the approach in~\cite{heylen2016alternatingAngleMinimization} employs an angle minimization strategy.  A significant reduction of the computational complexity is obtained for $P$ small and $C_{\cred{p}}$ possibly very large, with unmixing accuracy similar to MESMA.
Another approach formulates problem~\eqref{eq:mesma} as a mixed-integer optimization problem in order to benefit from advanced software packages~\cite{mhenni2018MESMA_MILP}.
However, none of these works consider the multitemporal formulation of MESMA.

MESMA has recently been applied to multitemporal SU problems such as monitoring of rainforests~\cite{somers2013invasiveHawaiiMultiTemporalBandWeighting,somers2013uncorrelatedBandSelectionInstabilityIndex,dudley2015multitemporalLibraryPhenologicalGradiantsMESMA} and shrublands~\cite{lippitt2018multidateMESMAshrublands}. Some of these methods improve MESMA performance by employing strategies such as band selection and weighting~\cite{somers2013uncorrelatedBandSelectionInstabilityIndex,somers2013invasiveHawaiiMultiTemporalBandWeighting}, or library construction from multiple time instants~\cite{dudley2015multitemporalLibraryPhenologicalGradiantsMESMA}. However, they do not explicitly explore temporal correlation between the abundance maps at adjacent time instants.
In the following, we propose a model for the evolution of the abundance maps over time, accounting for both small and large variations. This will allow us to devise an efficient algorithm to address multitemporal SU problems.

\section{Fast multitemporal MESMA} \label{sec:proposedMethod}

The multitemporal SU problem can be introduced with generic terms as follows: given a sequence of image pixels~$\{\by_{t,n}\}$, for $t=1,\ldots,T$ time instants and $n=1,\ldots,N$ pixels, and a spectral library~$\calM$ as defined in~\eqref{eq:specLibFull}, estimate the corresponding fractional abundances~$\{\ba_{t,n}\}$ and EM models~$\bM_{t,n}$. For simplicity, we shall assume that all the images are spatially aligned, such that for each $\cred{n}\in\{1,\ldots,N\}$ the pixels $\by_{t,\cred{n}}$, $t=1,\ldots,T$ refer to the same spatial location.
A simple solution to this problem would be to directly apply the techniques discussed in Section~\ref{sec:MESMA} to each pixel individually. This, however, ignores important temporal information contained in the image sequence which can be used in order to devise an efficient algorithm.

We propose to model the evolution of the abundance maps by considering its changes to be composed of a small additive signal, and of large sparse changes. That is, the observation model represents the pixels at time $t$ and $t+1$ as follows:
\begin{subequations} \label{eq:multit_model_i}
\begin{align}
	\by_{t,n} & {}={} \bM_{t,n} \ba_{t,n} + \be_{t,n},
	\label{eq:multit_model_ia}
    \\
    \by_{t+1,n} & {}={} \bM_{t+1,n} \big(\ba_{t,n} + \bdelta_{t,n} + \bs_{t,n}\big) + \be_{t+1,n}
    \label{eq:multit_model_ib}
    \\
    & {}={} \bM_{t+1,n} \ba_{t+1,n} + \be_{t+1,n},
    \nonumber 
\end{align}
\end{subequations}
where $\bM_{t,n}$ is the (true) endmember matrix for pixel $n$ at time instant $t$, and $\be_{t,n}$ is an additive noise vector. Changes taking place in the abundances between time $t$ and time $t+1$ are modeled as a combination of a small magnitude term $\bdelta_{t,n}$ and a spatially sparse, high magnitude term $\bs_{t,n}$, which represents abrupt variations taking place in a small number of image pixels. Note that, although model~\eqref{eq:multit_model_i} dictates the relationship between a single pair of images, the extension to multiple images is trivial and is thus omitted here for simplicity.

It turns out that the structure outlined in the model~\eqref{eq:multit_model_i} can be explored in order to devise an efficient MESMA-based SU algorithm.
We propose to use an online strategy to estimate the abundances and the EM matrices at time instant $t+1$ based on an estimate $\widehat{\ba}_{t,n}$ of the abundances at time instant~$t$. The procedure is a two-step one:
\begin{enumerate}
    \item Considering $\ba_{t,n}\equiv\widehat{\ba}_{t,n}$, estimate $\bM_{t+1,n}$, $\bdelta_{t,n}$ and $\bs_{t,n}$ that best represent pixel $\by_{t+1,n}$ in the model~\eqref{eq:multit_model_i};
    \item Set $\widehat{\ba}_{t+1,n}=\widehat{\ba}_{t,n}+\widehat{\bdelta}_{t,n}+\widehat{\bs}_{t,n}$ and repeat for the next image.
\end{enumerate}

Taking into consideration the prior information stated about the properties of $\bdelta_{t,n}$ and $\bs_{t,n}$, and the knowledge of the spectral library $\calM$, this can be translated 
\cred{into the following objectives and constraints for an optimization problem:
\begin{itemize}
    \item Minimize the reconstruction error w.r.t. $\bM\in\mathcal{M}$, $\bdelta_{t,n}$, and $\bs_{t,n}$, given by $\big\|\by_{t+1,n} - \bM \big(\widehat{\ba}_{t,n} + \bdelta_{t,n} + \bs_{t,n}\big)\big\|$;
    \item Preserve the nonnegativity and sum-to-one constraints on $\ba_{t+1,n}=\widehat{\ba}_{t,n}+\bdelta_{t,n}+\bs_{t,n}$;
    \item Consider that $\|\bdelta_{t,n}\|$ should be small and $\bs_{t,n}$ spatially sparse (i.e., nonzero only in a small number of pixels).
\end{itemize}}

\cred{Rather than accounting for all objectives outlined above at once by devising a single computationally demanding optimization problem,} we adopt an alternative strategy to obtain an efficient solution.
First, let us assume that $\bs_{t,n}=\cb{0}$. If $\bdelta_{t,n}$ is sufficiently small and $\widehat{\ba}_{t,n}$ is a good estimate of the true abundance $\ba_{t,n}$, we have:
\begin{align} \label{eq:approximation_a1}
\begin{split}
    \ba_{t,n}+\bdelta_{t,n} & {}\approx{} \ba_{t,n}
    \\
    & {}\approx{} \widehat{\ba}_{t,n}.
\end{split}
\end{align}
In that case, we can isolate the problem of estimating $\bM\in\calM$ \cred{from that of estimating $\bdelta_{t,n}$ and $\bs_{t,n}$} in order to solve it separately in a much simpler manner. We formulate the optimization problem as follows:
\begin{align} \label{eq:opt_appr_perm_i}
	RE_{t+1,n} {}={} & \min_{\bM\in\mathcal{M}} 
    \,\,\big\| \by_{t+1,n} - \bM\, \widehat{\ba}_{t,n} \big\| \,.
\end{align}
The endmember matrix $\widehat{\bM}_{t+1,n}$ obtained by solving problem~\eqref{eq:opt_appr_perm_i} can then be used to compute abundance vector $\widehat{\ba}_{t+1,n}$ with a single run of the FCLS algorithm.

However, this strategy relies on a strong hypothesis, namely, $\bs_{t,n}=\cb{0}$ and $\widehat{\ba}_{t,n}\approx\ba_{t,n}$, in order for the approximation in~\eqref{eq:approximation_a1} to hold. This hypothesis may not be satisfied for all pixels.
Fortunately, it turns out that we can devise a simple strategy to address those cases without significantly compromising the performance of the algorithm. Specifically, by evaluating the magnitude of the reconstruction error $RE_{t+1,n}$ in~\eqref{eq:opt_appr_perm_i}, we can indirectly identify if there were any significant changes in the abundance vector by testing whether $RE_{t+1,n}$ is larger than a given threshold $RE_0$ and, if so, estimate the corresponding abundance vector from scratch using MESMA \cred{or, alternatively, an algorithm such as AAM}.

Intuitively, the reason this works is that, if the spectral library $\calM$ is not too large, and if $\ba_{t+1,n}\approx\widehat{\ba}_{t,n}$ is not satisfied, then we cannot accurately reconstruct $\by_{t+1,n}$ from problem~\eqref{eq:opt_appr_perm_i}. In the next section, we shall formalize this intuition by providing a deeper theoretical analysis of this method. 
\cred{Setting the threshold $RE_0$ offers a trade-off between accuracy (which tends to MESMA's for small $RE_0$) and computational performance. We propose to choose it as:
\begin{align}
    RE_0 = \frac{K}{U} \sum_{\by\in\calU} \Big( \min_{\bM\in\calM} \,\mathop{\min}_{{\ba\geq\cb{0},\, \cb{1}^\top\ba=1}} \big\|\by-\bM\ba\big\| \Big) \,,
    \label{eq:threshold_param}
\end{align}
where $\calU=\{\by_1,\ldots,\by_U\}$ is a set with $U$ pixels and $K\in\amsmathbb{R}_+$. The pixels in $\calU$ should be similar to those contained in the image sequence $\{\by_{t,n}\}$ to be unmixed, so that $RE_0$ approaches $K$ times the average optimal reconstruction error of the data. The proportion $K$ controls how much the $RE_{t+1,n}$ of~\eqref{eq:opt_appr_perm_i} can deviate from the estimated optimal value before we decide an abrupt change occurred.
The inner optimization problem in~\eqref{eq:threshold_param} can be solved using MESMA or AAM.}
The complete procedure is detailed in Algorithm~\ref{alg:proposed_alg}.

\begin{algorithm} [thb]
\small
\SetKwInOut{Input}{Input}
\SetKwInOut{Output}{Output}
\caption{FM-MESMA Algorithm~\label{alg:proposed_alg}}
\Input{Multitemporal HS images $\{\by_{t,n}\}$, endmember library~$\calM$, threshold proportion parameter~$K$.}
Perform SU for the first HI $\by_{1,n}$, using the \cred{MESMA or AAM algorithms} to obtain $\widehat{\ba}_{1,n}$ \cred{and $\widehat{\bM}_{1,n}$,} for $n=1,\ldots,N$\;
Compute $RE_0$ \cred{according to~\eqref{eq:threshold_param}} \;
Initialize the change maps indicator function as $\widehat{\biota}_{\bs,t,n}=0$, $t=1,\ldots,T$, $n=1,\ldots,N$ \;
\For{$t=1,\ldots,T-1$}{
\For{$n=1,\ldots,N$}{
Solve problem~\eqref{eq:opt_appr_perm_i} to obtain the EM matrix~$\widehat{\bM}_{t+1,n}$ and the reconstruction error $RE_{t+1,n}$ \;
\uIf{$RE_{t+1,n}\leq RE_0$}{
Estimate $\widehat{\ba}_{t+1,n}$ using the FCLS method with $\widehat{\bM}_{t+1,n}$ as the EM matrix\;
}\Else{
Estimate $\widehat{\ba}_{t+1,n}$ and $\widehat{\bM}_{t+1,n}$ using MESMA \cred{or AAM and set} $\widehat{\biota}_{\bs,t,n}=1$ \;
}
}}
\KwRet \cred{Estimated abundances $\big\{\widehat{\ba}_{t,n}\big\}$, detected change maps $\big\{\widehat{\biota}_{\bs,t,n}\big\}$, and endmembers $\big\{\widehat{\bM}_{t,n}\big\}$} \;
\end{algorithm}

\section{Theoretical guarantees} \label{sec:theoremsAndAnalysis}

\cred{FM-MESMA (Algorithm~\ref{alg:proposed_alg}) relies on important assumptions (such as~\eqref{eq:approximation_a1}) in order to split SU into EM selection and abundance estimation, and to correctly detect abundance changes while maintaining a small complexity.
This raises questions regarding how the accuracy of the method is affected by the underlying assumptions and by the different variables involved in the model, such as: the noise $\be_{t,n}$, the small and abrupt abundance changes $\bdelta_{t,n}$ and $\bs_{t,n}$, the library $\calM$, and the accuracy of the estimated abundances at the previous time instants.
To investigate these questions, in this section we derive theoretical results in the form of two theorems, which provide conditions under which 1) the EM matrix can be correctly recovered from~\eqref{eq:opt_appr_perm_i}, and 2) the abrupt abundance changes $\bs_{t,n}$ can be correctly identified based on the reconstruction error $RE_{t+1,n}$.
}
\cred{Informally, the main findings related to each theorem can be summarized as:
\begin{enumerate}
    \item When $\bs_{t,n}=\cb{0}$, FM-MESMA recovers the correct EMs from~\eqref{eq:opt_appr_perm_i} if the pairwise difference between the signatures in each library $\calM_p$ is sufficiently large when compared to a measure related to the coherence between signatures at different libraries, to the noise $\be_{t,n}$ and to the abundance temporal variations $\bdelta_{t,n}$ (or to the errors in $\widehat{\ba}_{t,n}$).
    \item When an abrupt abundance change occurs, the value of $RE_{t+1,n}$ will be significantly larger than when $\bs_{t,n}=\cb{0}$ (thus making it easy to detect) as long as $\bs_{t,n}$ is sufficiently large compared to $\bdelta_{t,n}$, to $\be_{t,n}$, and to the maximum pairwise difference between the signatures of each EM in the library.
\end{enumerate}}
To proceed further, we shall assume that the true EM matrix $\bM_{t+1,n}$ in~\eqref{eq:multit_model_i} is an element of the library~$\calM$\footnote{Without loss of generality, one can always assume that $\bM_{t+1,n}\in\calM$ by incorporating any error $\Delta\bM_{t+1,n}$ into the additive noise term $\be_{t+1,n}$ \cred{by} adding $\Delta\bM_{t+1,n}\ba_{t+1,n}$.}.

\begin{theorem} \label{thm:theorem1}
Let us assume that $\widehat{\ba}_{t,n}=\ba_{t,n}$, $\bs_{t,n}=\cb{0}$ and that $\bM_{t+1,n}\in\calM$. Also assume that $\|\be_{t,n}\|<\Omega_e$, that $\max_{\bM\in\mathcal{M}}\|\bM\bdelta_{t,n}\|<\Omega_{\delta}$ and that \cred{the libraries $\calM_p$ satisfy:}
\begin{align}
    & \cred{\min_p \min_{\substack{\bm,\bm'\in\calM_p \\ \bm\neq\bm'}} \,\big\|\bm'-\bm\big\|^2 > \Omega_M}
    \label{eq:thm2_libCond1}
    \\
    & \cred{\max_{p,q} \max_{\substack{\bm,\bm'\in\calM_p \\ \widetilde{\bm},\widetilde{\bm}'\in\calM_q}} \, \big|\langle\bm'-\bm,\widetilde{\bm}'-\widetilde{\bm}\rangle\big| \leq \mu }
    \label{eq:thm2_libCond2}
\end{align}
\cred{Then, if $2\sqrt{P}(\Omega_e+\Omega_{\delta})<\sqrt{\Omega_M-(P-1)\mu}$} the solution of the optimization problem~\eqref{eq:opt_appr_perm_i} is~$\bM=\bM_{t+1,n}$.
\end{theorem}

\begin{proof}
Under these \cred{assumptions}, optimization problem~\eqref{eq:opt_appr_perm_i} can be written equivalently as
\begin{align} \label{eq:theorem1_deriv_i}
	& \min_{\bM\in\mathcal{M}} \,\, 
	\big\| \by_{t+1,n} - \bM\, \widehat{\ba}_{t,n} \big\|
	\nonumber \\
	={} & \min_{\bM\in\mathcal{M}} \,\, 
    \big\|\bM_{t+1,n}(\ba_{t,n}+\bdelta_{t,n})+\be_{t+1,n} - \bM \ba_{t,n}\big\|
    \\ \nonumber 
    ={} & \min_{\bM\in\mathcal{M}} \,\, 
    \big\|(\bM_{t+1,n}-\bM)\ba_{t,n}+\bM_{t+1,n}\bdelta_{t,n}+\be_{t+1,n}\big\| \,.
\end{align}
In order for the solution $\bM^*$ to problem~\eqref{eq:opt_appr_perm_i} to be the correct endmember model, we need the reconstruction error to be minimal only when $\bM_{t+1,n}-\bM^*=\cb{0}$. In that case, the reconstruction error is: $RE_{t+1,n}^*=\|\bM_{t+1,n}\bdelta_{t,n}+\be_{t+1,n}\|$. Equivalently, this means that $\forall\,\bM'\in\mathcal{M}, \,\bM'\neq\bM_{t+1,n}$, the following condition must be verified:
\begin{align} \label{eq:theorem1_deriv_ii}
    \big\| \bM_{t+1,n} & \bdelta_{t,n}+\be_{t+1,n} \big\|
    \\
    & {}<{} \big\| (\bM_{t+1,n}-\bM')\ba_{t,n}+\bM_{t+1,n}\bdelta_{t,n}+\be_{t+1,n} \big\|.
    \nonumber %
\end{align}

Using the Reverse Triangle Inequality with the r.h.s. of the above expression leads to:
\begin{align}
    \label{eq:theorem1_deriv_ii_2}
    & \big\| (\bM_{t+1,n}-\bM')\ba_{t,n} \big\| - \big\| \bM_{t+1,n}\bdelta_{t,n}+\be_{t+1,n} \big\|
    \\
    & \qquad \leq \big\| (\bM_{t+1,n}-\bM')\ba_{t,n}+\bM_{t+1,n}\bdelta_{t,n}+\be_{t+1,n} \big\| \,.
    \nonumber
\end{align}

Now, we can lower bound the l.h.s. of the previous expression. First, note that by hypothesis we have: 
\begin{align}
    \big\| \bM_{t+1,n}\bdelta_{t,n}+\be_{t+1,n} \big\| & {}\leq{} \|\bM_{t+1,n}\bdelta_{t,n}\|+\|\be_{t+1,n}\| 
    \nonumber \\
    & {}<{} \Omega_e+\Omega_{\delta} \,.
    \nonumber
\end{align}
\cred{Before proceeding to bound the first term in~\eqref{eq:theorem1_deriv_ii_2}, let us first derive the following auxiliary result concerning the matrix $\bM_{t+1,n}-\bM'$.}
\cred{
Since $L>P$, the Ger\v{s}hgorin circle theorem~\cite{horn2012matrixbook,richard2008online} can be used to derive a lower bound on the smallest singular value of $\bM_{t+1,n}-\bM'$ as follows:
\begin{align}
    \sigma_{\min}^2(\bM_{t+1,n}-\bM') 
    & = \lambda_{\min}\big(\bXi\big)
    \nonumber \\
    & \geq \min_p \,\, \Xi_{pp} - \sum_{j: j\neq p} |\Xi_{pj}|\,,
    \label{eq:gershgorin_i}
\end{align}
for $\bXi=\big(\bM_{t+1,n}-\bM'\big)^\top\big(\bM_{t+1,n}-\bM'\big)$ with $\Xi_{ij}$ being its $(i,j)$-th element, where we used the fact that $\bXi$ is diagonalizable. Functions $\sigma_{\min}(\cdot)$ and $\lambda_{\min}(\cdot)$ denote the smallest singular value and the smallest eigenvalue of a matrix, respectively. 
Using condition~\eqref{eq:thm2_libCond2} (i.e., the differences between signatures belonging to distinct EM libraries have low coherence), the second term in~\eqref{eq:gershgorin_i} satisfies
\begin{align}
    \sum_{j: j\neq p} |\Xi_{pj}| &= \sum_{j: j\neq p} \big|\langle\bm_{t+1,n}^p-(\bm')^p,  \bm_{t+1,n}^j-(\bm')^j\rangle\big|
    \nonumber\\
    & \leq (P-1)\mu \,,
\end{align}
where $\bm_{t+1,n}^p$ and $(\bm')^p$ denote the $p$-th column of $\bM_{t+1,n}$ and $\bM'$, respectively. This leads to the following lower bound of~\eqref{eq:gershgorin_i}:
\begin{align}
    \min_p \, \Xi_{pp} -\!\! \sum_{j: j\neq p} |\Xi_{pj}| & \geq \min_p \, \big\|\bm_{t+1,n}^p-(\bm')^p\big\|^2 - (P-1)\mu
    \nonumber \\
    & > \Omega_M - (P-1)\mu \,,
\end{align}
since $\min_p\big\|\bm_{t+1,n}^p-(\bm')^p\big\|^2>\Omega_M$ due to condition~\eqref{eq:thm2_libCond1}. Since $\Omega_M-(P-1)\mu>0$, by combining the above equation with~\eqref{eq:gershgorin_i} and taking the square root we can bound $\sigma_{\min}(\bM_{t+1,n}-\bM')$ as:
\begin{align}
    \sqrt{\Omega_M - (P-1)\mu} & < \sigma_{\min}(\bM_{t+1,n}-\bM') \,.
    \label{eq:bound_sigmamin_i}
\end{align}
Due to $\ba_{t,n}$ being confined to the unit simplex, its L$_2$ norm satisfies $(1/\sqrt{P})\leq\|\ba_{t,n}\|$. When combined with~\eqref{eq:bound_sigmamin_i}, this leads to:
\begin{align}
    (1/\sqrt{P}) \sqrt{\Omega_M - (P-1)\mu} & < (1/\sqrt{P}) \sigma_{\min}(\bM_{t+1,n}-\bM') 
    \nonumber \\
    & \leq \|\ba_{t,n}\| \sigma_{\min}(\bM_{t+1,n}-\bM')
    \nonumber\\
    &\leq \|(\bM_{t+1,n}-\bM')\ba_{t,n}\|.
\end{align}
This last equation constitutes the desired bound on the first term of~\eqref{eq:theorem1_deriv_ii_2}. Note that by hypothesis $2(\Omega_e+\Omega_{\delta})<(1/\sqrt{P}) \sqrt{\Omega_M-(P-1)\mu}$, which implies that:}
\begin{align}
    2\big\| \bM_{t+1,n}\bdelta_{t,n}+\be_{t+1,n} \big\| {}<{} \big\| (\bM_{t+1,n}-\bM')\ba_{t,n} \big\|
    \nonumber \,.
\end{align}
Finally, subtracting $\|\bM_{t+1,n}\bdelta_{t,n}+\be_{t+1,n}\|$ from both sides of the previous expression leads to:
\begin{align}
	\big\| \bM_{t+1,n} & \bdelta_{t,n}+\be_{t+1,n} \big\| 
    \nonumber \\
    {}<{} & \big\|(\bM_{t+1,n}-\bM')\ba_{t,n}\| - \|\bM_{t+1,n}\bdelta_{t,n}+\be_{t+1,n} \big\|
	\nonumber \\
    {}\leq{} & \big\|(\bM_{t+1,n}-\bM')\ba_{t,n}+\bM_{t+1,n}\bdelta_{t,n}+\be_{t+1,n} \big\| \,,
    \nonumber
\end{align}
for any $\bM'\in\mathcal{M}$, $\bM'\neq\bM_{t+1,n}$, which is precisely the necessary condition stated in equation~\eqref{eq:theorem1_deriv_ii}. Therefore, any $\bM'\neq\bM_{t+1,n}$ leads to a larger reconstruction error than using $\bM=\bM_{t+1,n}$, which is the optimal solution to~\eqref{eq:opt_appr_perm_i}.
\end{proof}

The rationale behind Theorem~\ref{thm:theorem1} is that, when the changes between the abundances $\ba_{t,n}$ and $\ba_{t+1,n}$ are small, then the variations observed between the pixel spectra $\by_{t,n}$ and $\by_{t+1,n}$ are due to 1) changes between the EM signatures $\bM_{t,n}$ and $\bM_{t+1,n}$, and 2) to the presence of additive noise. Thus, if the signal to noise ratio is high, which is common in HIs, and the EM matrices contained in the library $\calM$ are sufficiently different from one another, then the EM matrix in~$\calM$ that best reconstructs the new image $\by_{t+1,n}$ (with $\ba_{t+1,n}\approx\widehat{\ba}_{t,n}$) as the optimal solution to problem~\eqref{eq:opt_appr_perm_i} will be the correct EM matrix $\bM_{t+1,n}$, since the increase in reconstruction error due to the choice of any other matrix in $\calM$ will be greater than that due to the noise and abundance variations.

Although Theorem \ref{thm:theorem1} clarifies the circumstances under which the optimization problem~\eqref{eq:opt_appr_perm_i} provides the correct EM matrix, another important question is whether the detection strategy in \cred{FM-MESMA} can correctly identify the cases when it fails, so that the solution must then be processed using MESMA. Recall this mechanism allows the algorithm to deal with abrupt abundance variations.
It turns out that under similar assumptions to those made in Theorem~\ref{thm:theorem1}, if the EM matrices contained in~$\calM$ are not too different from one another, then abrupt abundance variations can be correctly identified by means of the behavior of the reconstruction error of problem~\eqref{eq:opt_appr_perm_i}, as stated in the following theorem.

\begin{theorem} \label{thm:theorem2}
Assume that $\widehat{\ba}_{t,n}=\ba_{t,n}$, that $\bs_{t,n}\neq\cb{0}$ and satisfies $\min_{\bM\in\mathcal{M}}\|\bM\bs_{t,n}\|>\Omega_s$ and that $\bM_{t+1,n}\in\calM$. Assume that $\|\be_{t,n}\|<\Omega_e$, that $\max_{\bM\in\mathcal{M}}\|\bM\bdelta_{t,n}\|<\Omega_{\delta}$ and that \cred{the maximum distances between signatures in each EM library is bounded as
\begin{align}
    \max_{p}\max_{\bm,\bm'\in\calM_p} \big\|\bm'-\bm\big\|<\Omega_M'.
\end{align}}
Then, if $\cred{\sqrt{P}\,\Omega_M'}+\Omega_{\delta}+\Omega_e<\Omega_s/(\cred{F}+1)$ the reconstruction error will be at least \cred{$F$} times larger than if $\bs_{t,n}=\cb{0}$.
\end{theorem}
\begin{proof}
Under these \cred{assumptions}, the optimization problem~\eqref{eq:opt_appr_perm_i} can be written equivalently as:
\begin{align}
	& \min_{\bM\in\mathcal{M}} \,\, 
	\big\| \by_{t+1,n} - \bM\, \widehat{\ba}_{t,n} \big\|
	\nonumber \\
	={} & \min_{\bM\in\mathcal{M}} \,\, 
    \big\| \bM_{t+1,n}(\ba_{t,n}+\bdelta_{t,n}+\bs_{t,n})+\be_{t+1,n} - \bM \ba_{t,n} \big\|
    \nonumber 
\end{align}
Using the Reverse Triangle Inequality, we have:
\begin{align} \label{eq:proof_s_i}
	& \!\!\! \big\| \bM_{t+1,n}\bs_{t,n} \big\| 
    - \big\| \bM_{t+1,n}(\ba_{t,n}+\bdelta_{t,n})+\be_{t+1,n} -\bM\ba_{t,n} \big\|
    \nonumber \\ 
    & \hspace{0.1cm} \leq{} \big\| \bM_{t+1,n}(\ba_{t,n}+\bdelta_{t,n}+\bs_{t,n})+\be_{t+1,n} - \bM \ba_{t,n} \big\|
\end{align}
for any solution $\bM\in\mathcal{M}$.

We can upper bound the second term in the l.h.s. of the above expression, that is, the reconstruction error without the effect of $\bs_{t,n}$, as follows:
\begin{align} \label{eq:proof_s_i_2}
    & \big\| \bM_{t+1,n}(\ba_{t,n}+\bdelta_{t,n})+\be_{t+1,n} -\bM\ba_{t,n} \big\|
    \\
    & \hspace{0.6cm} = \big\| (\bM_{t+1,n}-\bM)\ba_{t,n}+\bM_{t+1,n}\bdelta_{t,n}+\be_{t+1,n} \big\|
    \nonumber \\
    & \hspace{0.6cm} \leq \big\| (\bM_{t+1,n}-\bM)\ba_{t+1} \big\| + \big\| \bM_{t+1,n}\bdelta_{t,n} \big\| + \big\| \be_{t+1,n} \big\|.
    \nonumber %
\end{align}
\cred{Since $\ba_{t+1}$ is confined to the unit simplex, its L$_2$ norm satisfies
$\|\ba_{t+1}\|\leq1$, which allows us to bound the first term above using the properties of matrix norms as:
\begin{align}
    & \big\| (\bM_{t+1,n}-\bM)\ba_{t+1} \big\| \leq \big\|\bM_{t+1,n}-\bM\big\|\|\ba_{t+1}\| 
    \nonumber\\
    & \hspace{0.7cm} \leq \big\|\bM_{t+1,n}-\bM\big\|_F
    \nonumber\\
    & \hspace{0.7cm} \leq \max_p \,\,\sqrt{P} \big\|\bm_{t+1,n}^p - \bm^p\big\|
    \nonumber\\
    & \hspace{0.7cm} \leq \max_{p} \max_{\bm,\bm'\in\calM_p} \sqrt{P} \big\|\bm - \bm'\big\| < \sqrt{P}\,\Omega_M'
\end{align}
where $\bm_{t+1,n}^p$ and $\bm^p$ are the $p$-th columns of $\bM_{t+1,n}$ and $\bM$, respectively. Thus,~\eqref{eq:proof_s_i_2} can be bounded as:}
\begin{align} \label{eq:proof_s_ii}
    & \cred{\big\| \bM_{t+1,n}(\ba_{t,n}+\bdelta_{t,n})+\be_{t+1,n} -\bM\ba_{t,n} \big\|}
    \nonumber \\
    & \hspace{0.7cm}
    < \cred{\sqrt{P}\,\Omega_M'}+\Omega_{\delta}+\Omega_e
    \nonumber \\
    & \hspace{0.7cm} < \Omega_s \,(\cred{F}+1)^{-1}
    \nonumber \\
    & \hspace{0.7cm} < \big\| \bM_{t+1,n}\bs_{t,n} \big\| \,(\cred{F}+1)^{-1}
\end{align}
Now, by multiplying both sides of~\eqref{eq:proof_s_ii} by $\cred{F}+1$, subtracting $\|(\bM_{t+1,n}-\bM)\ba_{t,n}+\bM_{t+1,n}\bdelta_{t,n}+\be_{t+1,n}\|$ from each side, and finally using the result in~\eqref{eq:proof_s_i}, we obtain:
\begin{align}
    & \cred{F} \big\| (\bM_{t+1,n}-\bM)\ba_{t,n}+\bM_{t+1,n}\bdelta_{t,n}+\be_{t+1,n} \big\|
    \nonumber \\
    & \hspace{0.5cm} < \big\| \bM_{t+1,n}\bs_{t,n}\| 
    \nonumber \\
    & \hspace{1.2cm} - \|(\bM_{t+1,n}-\bM)\ba_{t,n}+\bM_{t+1,n}\bdelta_{t,n}+\be_{t+1,n} \big\|
    \nonumber \\
    & \hspace{0.5cm} \leq RE_{t+1,n}^*
    \nonumber \\
    & \hspace{0.5cm} \leq \big\| \bM_{t+1,n}(\ba_{t,n}+\bdelta_{t,n}+\bs_{t,n})+\be_{t+1,n} - \bM \ba_{t,n} \big\|
    \nonumber
\end{align}
for any solution $\bM\in\mathcal{M}$, where $RE_{t+1,n}^*$ is the optimal reconstruction error of problem~\eqref{eq:opt_appr_perm_i}. This result means that, if $\bs_{t,n}\neq\cb{0}$, the reconstruction error $RE_{t+1,n}^*$ will be at least $\cred{F}$ times larger than it would be if $\bs_{t,n}=\cb{0}$. 
\end{proof}

The rationale behind Theorem~\ref{thm:theorem2} is that, when there are large abundance changes between two consecutive time instants, that is, $\bs_{t,n}$ is large, then as long as the maximum pairwise difference between the elements in $\calM$ is not too large, it will not be possible to reconstruct the pixel $\by_{t+1,n}$ with good accuracy with any EM matrix in~$\calM$ because $\widehat{\ba}_{t,n}$ is too distant from $\ba_{t+1,n}$. This will lead to a large reconstruction error, what can be explored in order to identify such changes.

\cred{An important aspect of FM-MESMA is the assumption that an accurate estimate of the abundances is available at time $t$ (i.e., $\widehat{\ba}_{t,n}\approx\ba_{t,n}$). Due to the sequential nature of FM-MESMA, this approximation affects the results at time $t+1$. Since problem~\eqref{eq:opt_appr_perm_i} actually relies on the approximation $\widehat{\ba}_{t,n}\approx\ba_{t+1,n}$ to perform EM selection, errors in $\widehat{\ba}_{t,n}$ will affect it in the same way as $\bdelta_{t,n}$ (or, in the case of large errors, as $\bs_{t,n}$). Thus, while errors in $\widehat{\ba}_{t,n}$ can impact the performance of the method negatively, their effect can be controlled: if the errors are large and $K$ is properly selected, they will have the same effect as sudden changes, and cause the algorithm to reprocess the given pixel with MESMA or AAM.}

\subsection{Computational complexity analysis}

In this section we show how the computational complexity of the proposed algorithm compares to that of MESMA \cred{and AAM}. \cred{Before proceeding, let us denote by $\kappa\in[0,1]$ the average proportion of pixels that undergo changes between consecutive time instants $t$ and $t+1$.}

The operations in \cred{FM-MESMA} at each time instant $t$ with non-negligible computational complexity consist of 1) solving of optimization problem~\eqref{eq:opt_appr_perm_i} for all $N$ image pixels, and 2) running \cred{MESMA or AAM} to estimate the abundances of pixels that did undergo significant changes. 
For the first step, a simple enumeration strategy can be employed to solve problem~\eqref{eq:opt_appr_perm_i} by testing all possible EM models, which results in a complexity of about $\mathcal{O}(\prod_{\cred{p}=1}^P C_{\cred{p}} LP)$. %
For the second step, the \cred{complexity will} be on average $\kappa$ times \cred{the one of MESMA or of AAM, which will be discussed below}.

Determining the computational complexity of MESMA is less direct, since the optimization problem~\eqref{eq:mesma} does not have a closed form solution and depends on iterative algorithms. For simplicity, we assume that the inner FCLS problem in MESMA~\eqref{eq:mesma} is solved by approximately translating it into a nonnegative least squares (NNLS) problem of dimensions $(L+1)\times (P+1)$~\cite{heinz2001FCLS_unmixing}.
Different methods have been proposed to solve NNLS problems, including the active set methods, the interior point method, and other iterative approaches~\cite{chen2010methodsNNLS}. Although iterative approaches perform better in large scale problems, the interior point method works well for problems like MESMA and can give us an idea of its complexity~\cite{polyak2015projectedGradientNNLS}.

The interior point method needs $\mathcal{O}(\ln(\epsilon^{-1}))$ iterations, each with complexity $(P+1)^3$, in order to achieve a reconstruction error that is $\epsilon$-close to the global optimum~\cite[p. 393]{nesterov2018lecturesConvexOptimization}, i.e.,
\begin{align}
    \|\by_n - \bM\widehat{\ba}_n\|^2 - \|\by_n - \bM\ba_n^*\|^2 \leq \epsilon
\end{align}
where $\ba_n^*$ is abundance vector that minimizes the NNLS problem for a given~$\bM$.
Thus, since this problem is solved for each EM model, the computational complexity of MESMA is approximately $\mathcal{O}(\prod_{\cred{p}=1}^P C_{\cred{p}} P^3 \ln(\epsilon^{-1}))$.%

It can be seen that although the complexity of the proposed method and MESMA scale similarly as the size of the libraries~$C_{\cred{p}}$ increases, as long as $\kappa$ is small the proposed algorithm can handle scenarios with moderately large~$P$ much more easily since problem~\eqref{eq:opt_appr_perm_i} scales linearly with~$P$. 
This characteristic needs to be emphasized because, while several methods can effectively address the problem of reducing the size $C_{\cred{p}}$ of the libraries by removing redundant signatures~\cite{roth2012libraryPruningComparisoncost,dennison2003libraryPruningMESMAcost,dennison2004comparisonErrorMetricsEndmemberSelection}, a larger number $P$ of EM classes cannot be so easily circumvented. \cred{The AAM algorithm~\cite{heylen2016alternatingAngleMinimization}, for instance, requiring $\mathcal{O}(P2^P(L^3+PL\max_{p}C_{p}))$ operations at every iteration, has a larger base cost depending on terms such as $L^3$ and scales quickly with $P$, making it costly in scenarios where $P$ is large and $C_p$ is small.} 
\cred{Such cases in which the proposed algorithm is particularly faster (i.e., small $C_p$) are of special practical interest. This is because recently proposed state of the art approaches to spectral library reduction for MESMA have reported experimental results indicating that libraries could be reduced to between two and five (averaging three) signatures per EM without an appreciable drop in performance~\cite{meerdink2019bundleExtractionPartiallyLabelled}.}
\cred{Finally, we note that the best choice, in terms of complexity, among using MESMA or AAM to unmix the changed pixels in the proposed algorithm may depend on each scenario. If $P$ or $C_p$ is small, MESMA can be a good choice due to its smaller base cost compared to AAM, whereas for moderate $C_p$ AAM will perform faster.}

\section{Experimental Results} \label{sec:experimentalResults}

We shall now evaluate the performance of the proposed \cred{FM-MESMA} algorithm using simulations with synthetic, semi-real, and real data. Our method is compared with MESMA and AAM~\cite{heylen2016alternatingAngleMinimization}, which are both library-based methods, with fully constrained least squares (FCLS) algorithm, and with the online unmixing algorithm (OU)~\cite{Thouvenin_IEEE_TIP_2016}.
\cred{The OU algorithm estimates both the abundances and one set of EMs for each time instant blindly from the HI using a two-stage stochastic optimization procedure. The EMs are modelled in OU as temporally smooth additive perturbation over a mean EM matrix.} \cred{FM-MESMA is implemented using MESMA to unmix the significantly changed pixels in step~10.}
In all simulations, the endmembers for the FCLS were extracted from the HI using the VCA algorithm~\cite{Nascimento2005}.
\cred{For FM-MESMA, we computed $RE_0$ in \eqref{eq:threshold_param} using the pixels in the HI at the initial time instant $t=1$, i.e., $\calU=\{\by_{1,1},\ldots,\by_{1,N}\}$.}
To evaluate the performance of the algorithms, different metrics were considered depending on the simulation setups. Performance metrics that are specific to the simulations with synthetic data will be defined in Section~\ref{sec:results_syntheticData}. For the simulations with semi-real and real data, we considered as metrics the root mean squared error (RMSE) and the spectral angle mapper (SAM). The RMSE between two sequences of matrices~$\bX_t$ and~$\bX_t^*$, for $t=1,\ldots,T$, is defined as:
\begin{align} %
    \text{RMSE}_{\bX} = \sum_{t=1}^T \sqrt{\frac{1}{T\,N_{\!\bX}}\|\bX_{\!t} - \bX_{\!t}^*\|^2_F}
    \nonumber
\end{align}
where $N_{\!\bX}$ is the number of elements in~$\bX_t$. The SAM between the true and estimated endmembers is defined as:
\begin{align}
    \text{SAM}_{\bM} = \frac{1}{TNP}\sum_{t=1}^{T} \sum_{n=1}^{N} \sum_{\cred{p}=1}^{P} \arccos\bigg(\frac{(\bm_{t,n}^{\cred{p}})^\top\widehat{\bm}_{t,n}^{\cred{p}}}{\|\bm_{t,n}^{\cred{p}}\|\|\widehat{\bm}_{t,n}^{\cred{p}}\|}\bigg),
    \nonumber
\end{align}
where $\bm_{t,n}^{\cred{p}}$ and $\widehat{\bm}_{t,n}^{\cred{p}}$ are the $\cred{p}$-th columns of the true and the estimated endmembers, respectively.

\subsection{Synthetic Data} \label{sec:results_syntheticData}

The simulations with synthetic data were designed to illustrate how \cred{FM-MESMA} performs when compared to MESMA and AAM \cred{in four different ways, namely:
\begin{enumerate}
    \item Its computational cost for different values of $P$ and $C_p$ (Section~\ref{sec:results_syntheticData_compCost});
    \item Its accuracy when detecting abrupt abundance changes, and the effect of the proportion of changed pixels (i.e., $\kappa$) on the computational cost of FM-MESMA (Section~\ref{sec:results_syntheticData_CDcomparison});
    \item Its accuracy when recovering the EMs from $\calM$ for different amounts of abundance temporal variation $\bdelta_{t,n}$ and different signal to noise ratios (SNRs) (Section~\ref{sec:results_syntheticData_libVar});
    \item Its accuracy when recovering the EMs from $\calM$ for different amounts of library variance $\sigma_{\calM}^2$ and different SNRs (Section~\ref{sec:results_syntheticData_abVar}).
\end{enumerate}
These different experiments provide an empirical assessment of some of the theoretical results derived in Section~\ref{sec:theoremsAndAnalysis}.}
For the general simulation setup, we considered sequences of $T=11$ images with $L=200$ bands and $N=1000$ pixels. Each pixel $\by_{t,n}$ was generated according to the LMM in~\eqref{eq:multit_model_ia}, where the true endmember matrix $\bM_{t,n}$ was sampled uniformly from a library $\calM$, and $\be_{t,n}$ was a white Gaussian noise. 
The EM library $\calM$ was generated randomly. The mean $\bmu_p$, $p=1,\ldots,P$ of each material was first sampled from a uniform distribution over interval $[0,1]^L$. Then each EM signature in $\calM_p$, $p=1,\ldots,P$ was generated as a sample from an isotropic Gaussian distribution $\calN(\bmu_p,\sigma_{\calM}^2\bI)$ and truncated in the interval $[0,1]^L$.
The abundances $\ba_{1,n}$ were sampled from a Dirichlet distribution. \cred{Between} each pair of images at instants $t$ and $t+1$, a proportion $\kappa$ of the $N$ pixels was changed with new samples, whereas the remaining ones were kept constant unless otherwise specified.
The other parameters such as $K$, the \cred{SNR}, the change ratio $\kappa$ and the library variance $\sigma_{\calM}^2$ will be specified in the following for each experiment. 
\cred{In general, the library variance is defined as:
\begin{align}
    \sigma_{\calM}^2 = \frac{1}{LP} \sum_{p=1}^P \tr\big\{\cov(\bm_p,\bm_p)\big\} \,,
\end{align}
where $\tr\{\cdot\}$ is the matrix trace operator and $\cov(\bm_p,\bm_p)$ denotes the covariance matrix of the signatures from the $p$-th EM, which can be estimated using the samples in the library~$\calM_p$ if unknown.}

\begin{table*}[htb]
    \centering
    \footnotesize
    \renewcommand{\arraystretch}{1.05}
    \caption{Execution \cred{times} for different values of $P$ and $C_{\cred{p}}$ (best results marked in bold, best results by our method marked in red).}
    \label{tab:comparative_synthetic_execution_times}
    \vspace{-0.2cm}
    \begin{tabular}{c|c|cccccccccccccccccccccccccccccccccccc} \hline
	&	\diagbox[height=14pt]{Method}{$C_{\cred{p}}$}    & 	2 &	3	&	4	&	5	&	6	&	7	&	8	&	9	&	10	\\ \hline\hline
																					
$P=2$	&	MESMA	&	1.20	&	1.21	&	1.27	&	1.43	&	1.37	&	1.57	&	1.60	&	1.66	&	1.96	\\
	&	AAM	&	6.61	&	6.48	&	6.67	&	7.11	&	6.79	&	6.88	&	6.98	&	7.06	&	7.17	\\
	&	\cred{FM-MESMA}	&	\cmark{\bf 1.19}	&	\cmark{\bf 1.16}	&	\cmark{\bf 1.18}	&	\cmark{\bf 1.30}	&	\cmark{\bf 1.21}	&	\cmark{\bf 1.31}	&	\cmark{\bf 1.31}	&	\cmark{\bf 1.35}	&	\cmark{\bf 1.55}	\\\hline
																					
$P=3$	&	MESMA	&	1.55	&	1.78	&	2.32	&	2.88	&	4.52	&	6.26	&	6.79	&	9.46	&	15.90	\\
	&	AAM	&	30.73	&	30.41	&	31.62	&	32.12	&	32.04	&	32.62	&	33.15	&	33.82	&	33.48	\\
	&	\cred{FM-MESMA}	&	\cmark{\bf 1.44}	&	\cmark{\bf 1.52}	&	\cmark{\bf 1.60}	&	\cmark{\bf 1.75}	&	\cmark{\bf 2.08}	&	\cmark{\bf 2.45}	&	\cmark{\bf 3.05}	&	\cmark{\bf 4.03}	&	\cmark{\bf 4.93}	\\\hline
																					
$P=4$	&	MESMA	&	2.91	&	4.97	&	11.04	&	36.74	&	52.37	&	103.38	&	180.75	&	282.21	&	490.92	\\
	&	AAM	&	98.53	&	98.12	&	101.76	&	103.84	&	103.99	&	106.57	&	107.61	&	108.82	&	109.05	\\
	&	\cred{FM-MESMA}	&	\cmark{\bf 2.14}	&	\cmark{\bf 2.30}	&	\cmark{\bf 3.07}	&	\cmark{\bf 5.17}	&	\cmark{\bf 6.88}	&	\cmark{\bf 11.58}	&	\cmark{\bf 16.79}	&	\cmark{\bf 28.87}	&	\cmark{\bf 45.57}	\\\hline
																					
$P=5$	&	MESMA	&	5.04	&	24.34	&	74.94	&	272.21	&	680.85	&	1482.15	&	2780.66	&	4882.36	&	7783.21	\\
	&	AAM	&	272.22	&	270.72	&	282.31	&	285.31	&	289.26	&	292.36	&	296.11	&	{\bf 295.97}	&	{\bf 300.89}	\\
	&	\cred{FM-MESMA}	&	\cmark{\bf 2.68}	&	\cmark{\bf 4.21}	&	\cmark{\bf 8.80}	&	\cmark{\bf 25.38}	&	\cmark{\bf 50.61}	&	\cmark{\bf 110.60}	&	\cmark{\bf 234.83}	&	383.27	&	645.80	\\\hline
																					
$P=6$	&	MESMA	&	8.61	&	81.69	&	398.69	&	1600.86	&	4625.57	&	12153.74	&	24423.06	&	$\infty$	&	$\infty$	\\
	&	AAM	&	678.18	&	691.35	&	715.88	&	728.04	&	745.54	&	{\bf 751.10}	&	{\bf 761.53}	&	{\bf 760.70}	&	{\bf 760.47}	\\
	&	\cred{FM-MESMA}	&	\cmark{\bf 3.66}	&	\cmark{\bf 9.30}	&	\cmark{\bf 36.11}	&	\cmark{\bf 135.68}	&	\cmark{\bf 405.91}	&	1069.17	&	2614.65	&	4707.32	&	8513.08	\\\hline
																					
$P=7$	&	MESMA	&	19.58	&	257.71	&	1933.77	&	8554.18	&	31760.77	&	$\infty$	&	$\infty$	&	$\infty$	&	$\infty$	\\
	&	AAM	&	1653.51	&	1671.48	&	1726.08	&	1753.13	&	{\bf 1788.30}	&	{\bf 1808.63}	&	{\bf 1829.27}	&	{\bf 1843.13}	&	{\bf 1861.54}	\\
	&	\cred{FM-MESMA}	&	\cmark{\bf 5.07}	&	\cmark{\bf 23.60}	&	\cmark{\bf 158.68}	&	\cmark{\bf 736.29}	&	2660.14	&	7728.61	&	21816.81	&	$\infty$	&	$\infty$	\\\hline
																					
$P=8$	&	MESMA	&	36.29	&	748.21	&	8394.23	&	49221.23	&	$\infty$	&	$\infty$	&	$\infty$	&	$\infty$	&	$\infty$	\\
	&	AAM	&	3876.57	&	3932.96	&	4079.04	&	4112.40	&	{\bf 4192.04}	&	{\bf 4368.23}	&	{\bf 4406.82}	&	{\bf 4453.85}	&	{\bf 4506.10}	\\
	&	\cred{FM-MESMA}	&	\cmark{\bf 6.50}	&	\cmark{\bf 68.31}	&	\cmark{\bf 656.97}	&	\cmark{\bf 3869.42}	&	17330.46	&	$\infty$	&	$\infty$	&	$\infty$	&	$\infty$	\\\hline
																					
$P=9$	&	MESMA	&	75.84	&	3167.33	&	38765.51	&	$\infty$	&	$\infty$	&	$\infty$	&	$\infty$	&	$\infty$	&	$\infty$	\\
	&	AAM	&	8972.31	&	9126.39	&	9354.76	&	{\bf 9505.15}	&	{\bf 9694.71}	&	{\bf 9861.76}	&	{\bf 9960.79}	&	{\bf 10041.00}	&	{\bf 10178.05}	\\
	&	\cred{FM-MESMA}	&	\cmark{\bf 9.78}	&	\cmark{\bf 211.21}	&	\cmark{\bf 2913.83}	&	19922.95	&	$\infty$	&	$\infty$	&	$\infty$	&	$\infty$	&	$\infty$	\\ \hline
    \end{tabular}
\end{table*}

\subsubsection{\textbf{Computational complexity analysis}}
\label{sec:results_syntheticData_compCost}

We evaluated the execution time of MESMA, AAM and \cred{FM-MESMA} for different values of $P\in\{2,\ldots,9\}$ and $C_{\cred{p}}\in\{2,\ldots,10\}$, $\cred{p}=1,\ldots,P$. We considered \cred{an} SNR of 40dB, a change ratio of $\kappa=0.01$, a library variance of $\sigma_{\calM}^2=0.12$ and $K=10$ in Algorithm~\ref{alg:proposed_alg}. The results are shown in Table~\ref{tab:comparative_synthetic_execution_times}. It can be seen that \cred{FM-MESMA} has a significantly smaller execution time when compared to MESMA. Moreover, there are also significant improvements over AAM when $P$ is large and $C_{\cred{p}}$ is small or moderate. Note that although AAM performs better when $C_{\cred{p}}$ is large, this situation is avoided in practice without significant impact in the performance by removing some redundant signatures from the library~\cite{roth2012libraryPruningComparisoncost,dennison2003libraryPruningMESMAcost,dennison2004comparisonErrorMetricsEndmemberSelection,roberts2003countBasedEndmemberSelection}. However, this is not the case for larger $P$, \cred{in which case FM-MESMA} leads to a performance improvement. \cred{We} solved~\eqref{eq:opt_appr_perm_i} using an exhaustive search procedure. More efficient solutions will be devised in the future.

\begin{figure}
    \centering
    \includegraphics[width=0.72\linewidth]{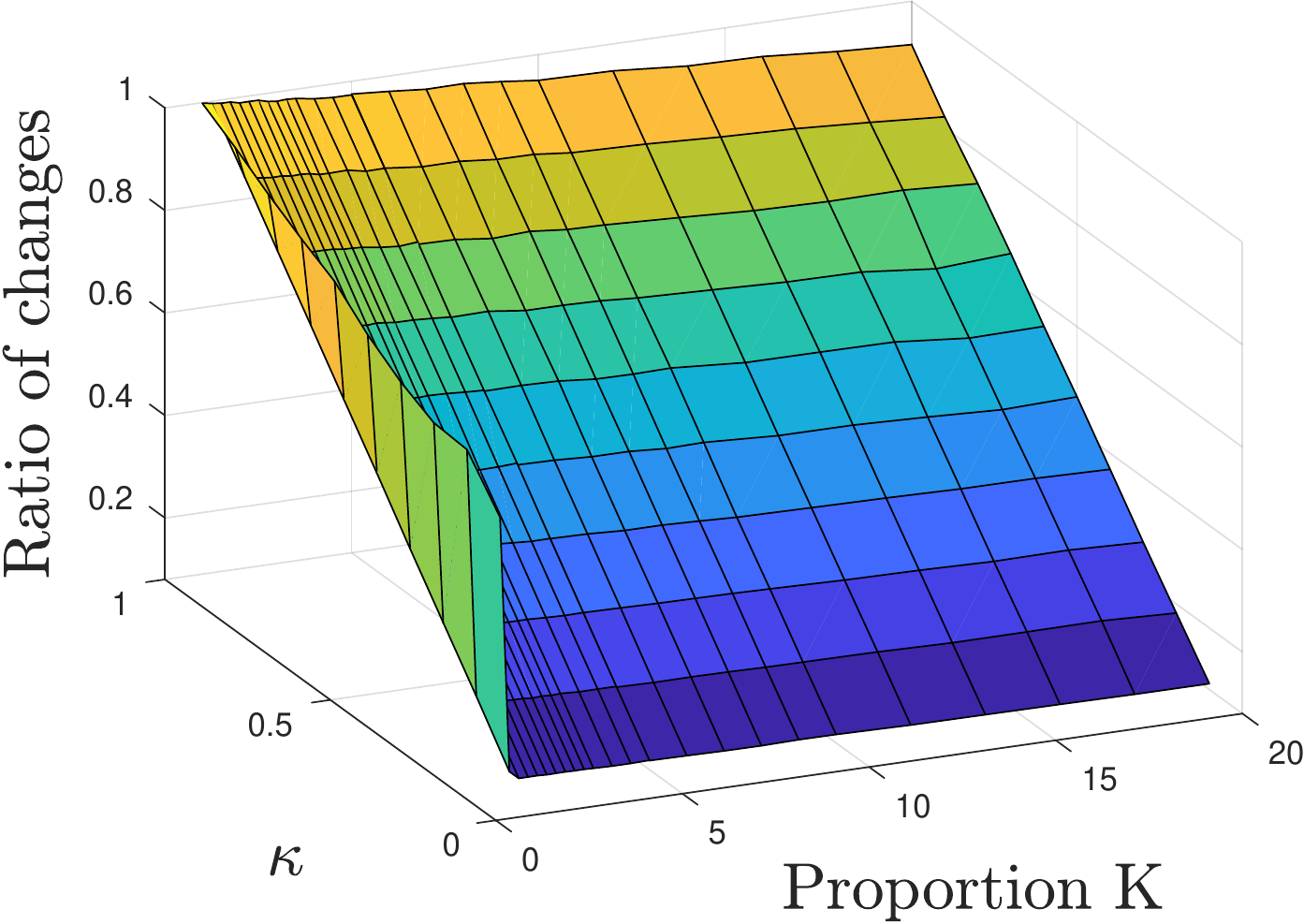}
    
    \medskip
    
    \includegraphics[width=0.7\linewidth]{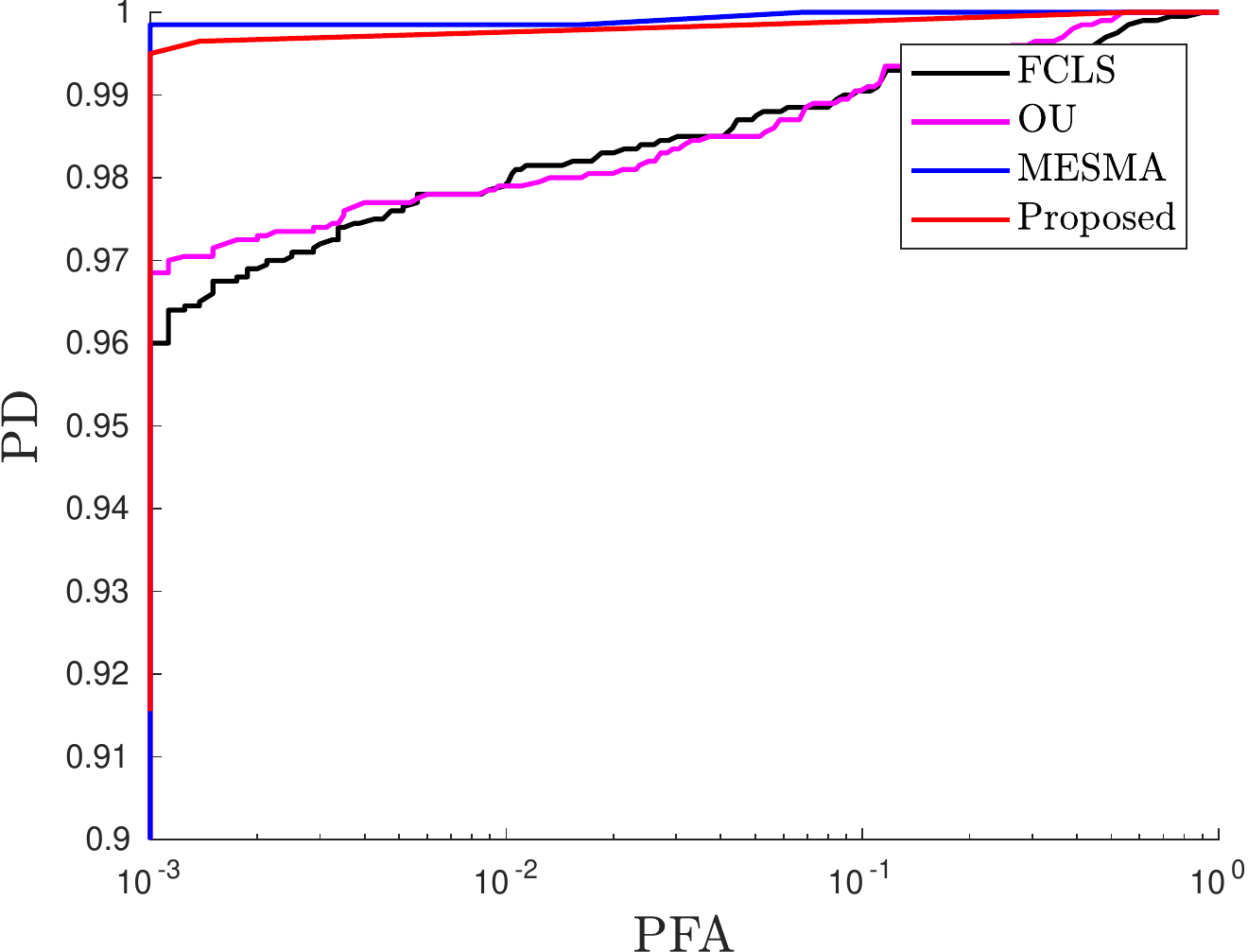}
    \caption{\cred{Large changes analysis: Number of detected changes as a function of $K$ and $\kappa$ (top), and change detection performance for $\kappa=0.2$ (bottom) (AAM results were omitted due to being similar to MESMA's).}}
    \label{fig:synthResults_largeAbundanceChanges_rev}
\end{figure}

\subsubsection{\textbf{Large change detection analysis}}
\label{sec:results_syntheticData_CDcomparison}

We evaluated the accuracy of \cred{FM-MESMA when detecting} pixels containing large abundance changes to reprocess them using the MESMA \cred{or AAM algorithms}. This experiment also allowed us to check numerically the theoretical results in Theorem~\ref{thm:theorem2}. \cred{The experiment is divided in two parts. First, we evaluate the effect of $K$ and $\kappa$ on the computational burden of the proposed method by measuring the proportion of pixels marked as changes depending on these variables. 
Afterwards, we compare the change detection accuracy of FM-MESMA to that of the other algorithms (where change detection was performed using the strategy in~\cite{erturk2015simplesSU_CD}) by fixing $\kappa=0.2$ and evaluating the probability of detection (PD) against the probability of false alarm (PFA), where PD and PFA are defined as:
\begin{align}
    \text{PD}_A &= %
    \sum_{t=2}^T \frac{ \sum_{n=1}^N \cred{\chi_{\mathbb{0}}}(\ba_{t,n}-\ba_{t-1,n}) \cdot \widehat{\biota}_{\bs,t,n}}{(T-1)\sum_{n=1}^N \cred{\chi_{\mathbb{0}}}(\ba_{t,n}-\ba_{t-1,n})}
    \nonumber
    \\
    \text{PFA}_A &= %
    \sum_{t=2}^T \frac{\sum_{n=1}^N \max\big\{\widehat{\biota}_{\bs,t,n} - \cred{\chi_{\mathbb{0}}} {(\ba_{t,n}-\ba_{t-1,n})},0\big\}}{(T-1)\big(N-\sum_{n=1}^N \cred{\chi_{\mathbb{0}}}(\ba_{t,n}-\ba_{t-1,n})\big)}
    \nonumber
\end{align}
where \cred{$\chi_{\mathbb{0}}(\cdot)$ is the indicator function of the set $\mathbb{0}=\{\cb{0}\}$ (i.e., $\chi_{\mathbb{0}}(\bx)=1$ if $\bx=\cb{0}$ and~$0$ otherwise)}. For both cases, we considered an SNR of 30dB, a library variance of $\sigma_{\calM}^2=0.12$, $P=4$ EMs and $C_p=3$. 
The results are shown in Figure~\ref{fig:synthResults_largeAbundanceChanges_rev}. It can be seen that when $K$ was not too close to one, the number of times a pixel had to be reprocessed depended mostly on the actual amount of changed pixels in the scene $\kappa$. 
In terms of accuracy, the change detection performance was satisfactory for all methods, with MESMA showing the best performance followed closely by FM-MESMA, and OU and FCLS being slightly worse. 
This suggests that one can select a moderate value of $K$ and obtain a good change detection accuracy without increasing the complexity unnecessarily.
However, devising a strategy to select an optimal value for $K$ is more complex and will be left as a subject for future work.
}

\begin{figure*}
    \centering
    \includegraphics[width=0.325\linewidth]{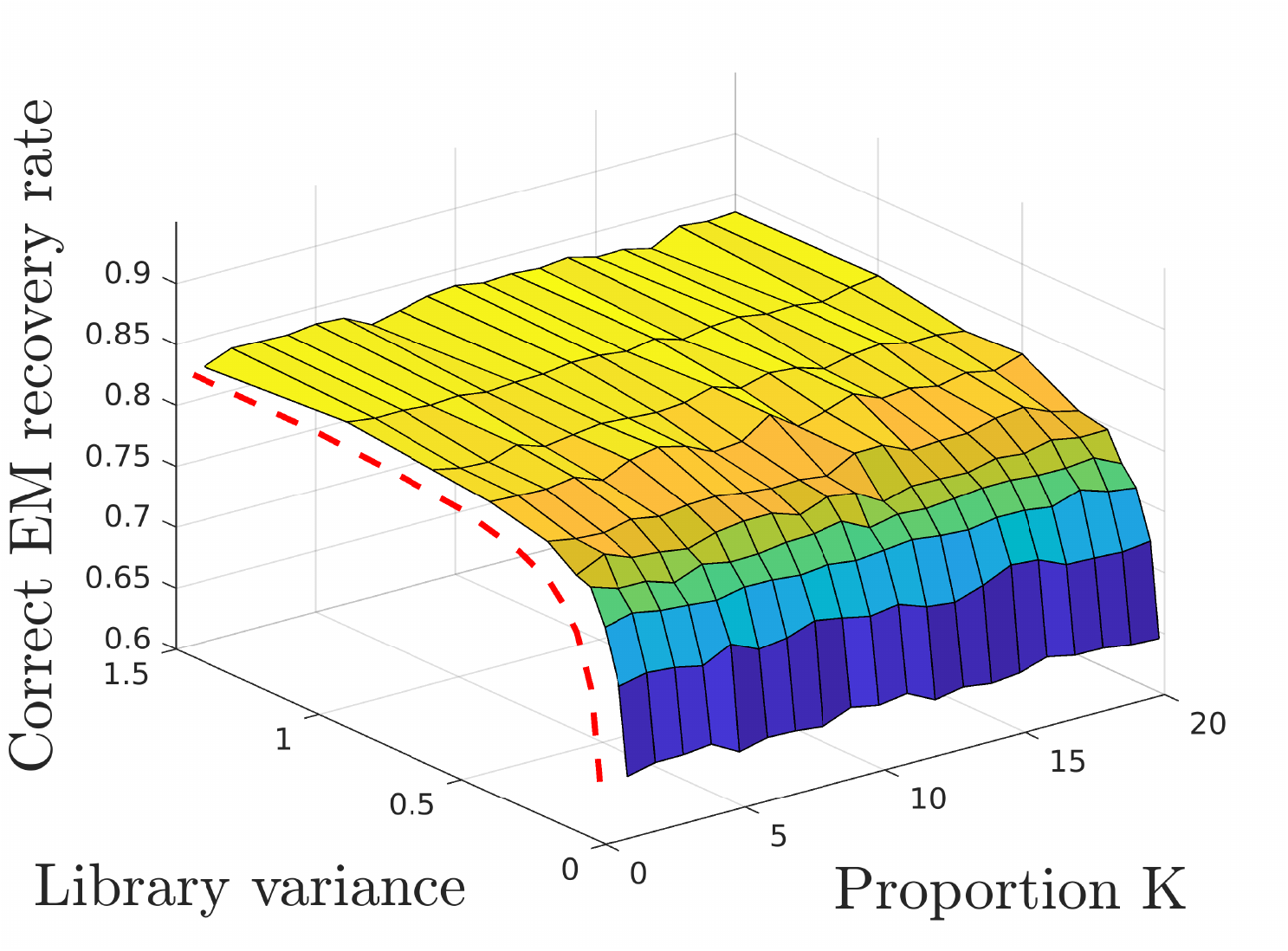}
    \includegraphics[width=0.325\linewidth]{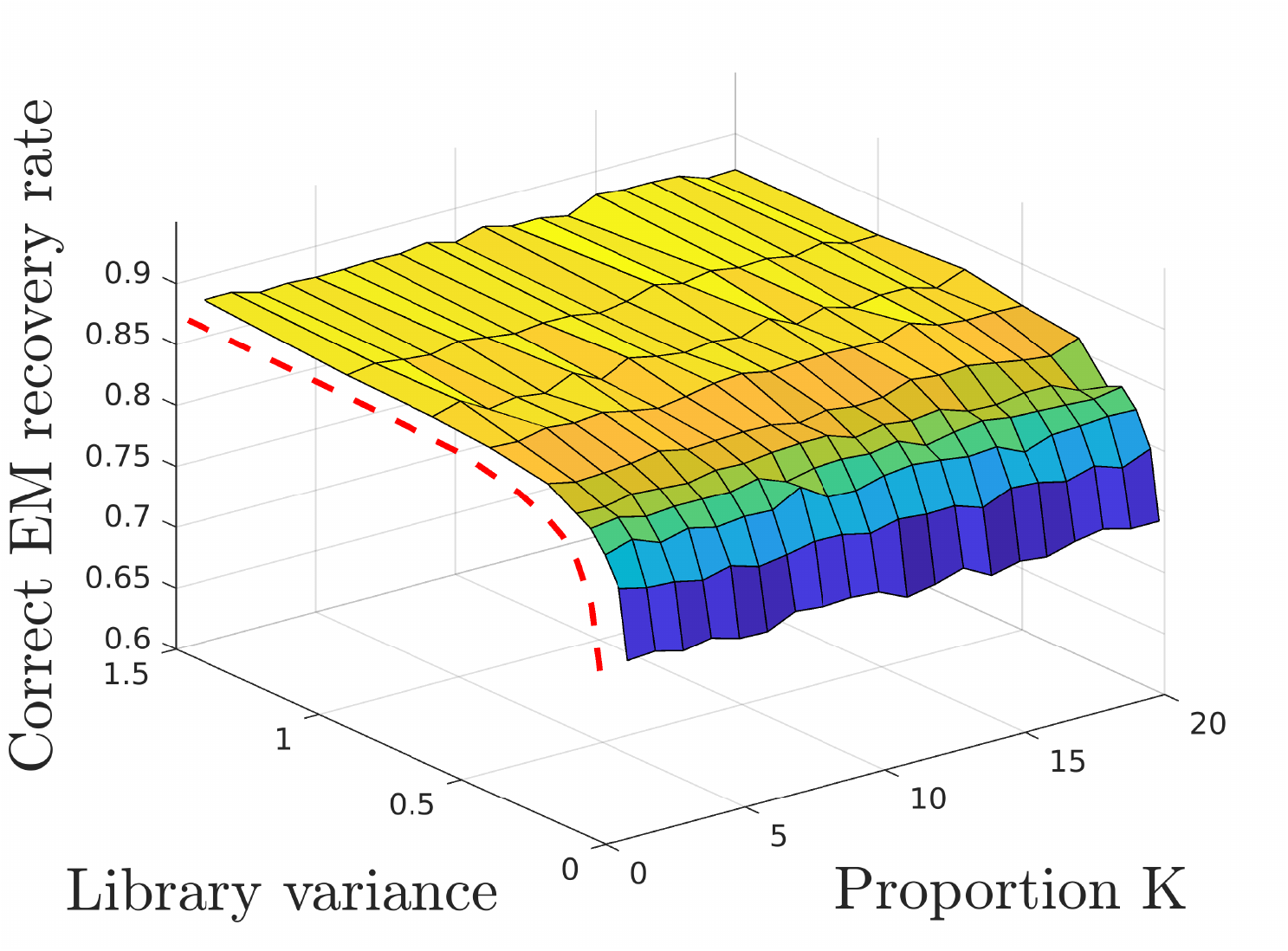}
    \includegraphics[width=0.325\linewidth]{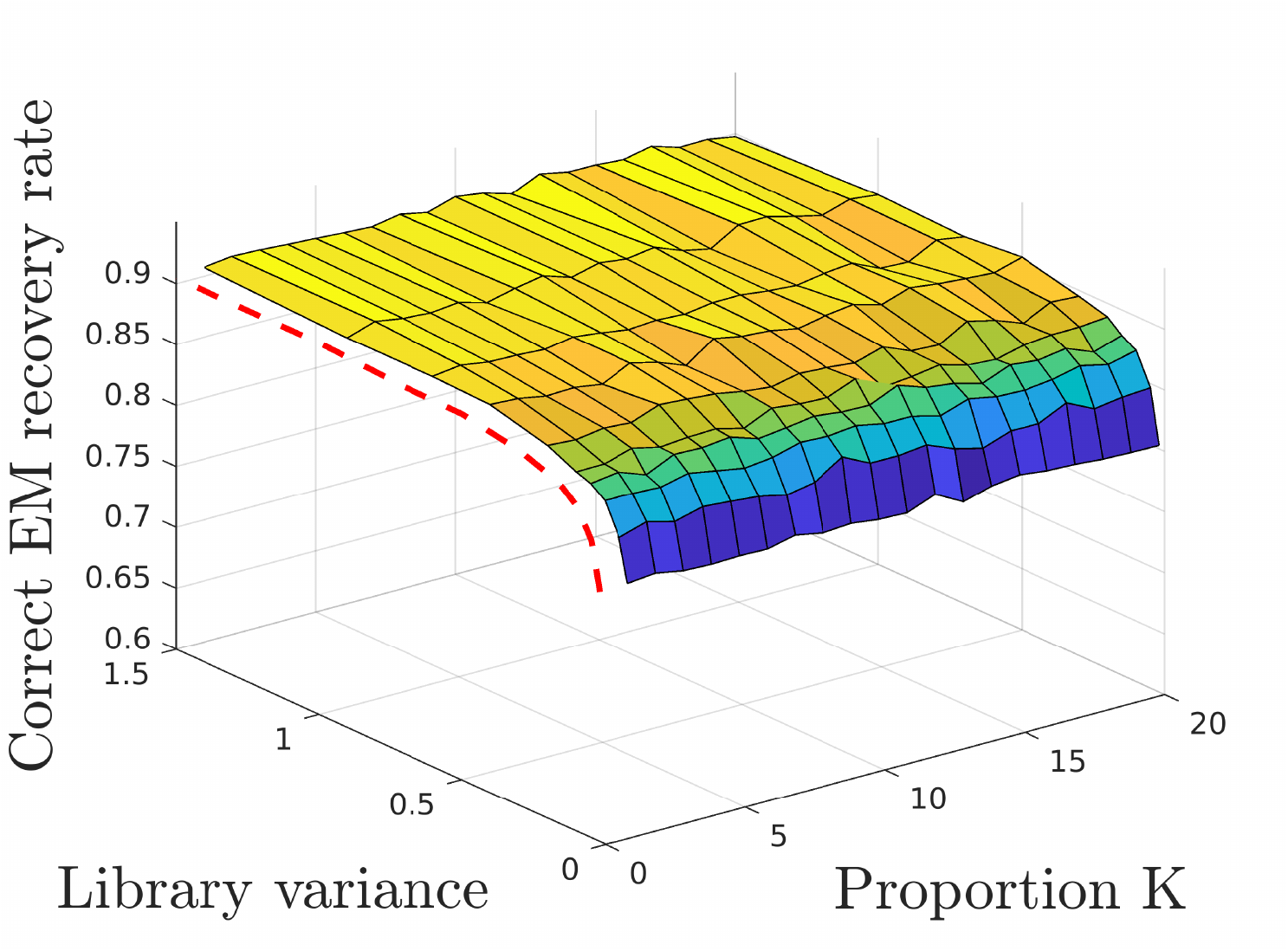}
    \caption{Library variance analysis: Correct EM recovery rate as a function of \cred{$\sigma_{\calM}^2$ and $K$}, for SNRs of 25dB (left), 35dB (middle) and 45dB (right).}
    \label{fig:synthResults_libraryVarianceAnalysis}
\end{figure*}
\subsubsection{\textbf{Library variance analysis}}  \label{sec:results_syntheticData_libVar}

We evaluated how often the solution of problem~\eqref{eq:opt_appr_perm_i} was the same as the true EM matrix $\bM_{t,n}$ when compared to that of MESMA, for different library variances.
This experiment also allowed us to validate numerically one part of Theorem~\ref{thm:theorem1}. We considered $\kappa=0.05$, $P=4$, $C_{\cred{p}}=3$, and different values of $K\in\{1,\ldots,20\}$ and $\sigma_{\calM}^2=\{0.02,0.05,0.1,0.15,0.2,0.3,0.5,0.7,1,1.5\}$. The performance was evaluated using the endmember positive predictive value (PPV)
\begin{align} \label{eq:synthMetrics_PPV_M}
    \text{PPV}_M &= \frac{1}{T} \sum_{t=1}^T \frac{ \sum_{n=1}^N \cred{\chi_{\mathbb{0}}} (\bM_{t,n}-\widehat{\bM}_{t,n})}{N}.
\end{align}
where $\bM_{t,n}\in\calM$ and $\widehat{\bM}_{t,n}\in\calM$ are the true and estimated EM matrices, respectively.

The results are shown in Figure~\ref{fig:synthResults_libraryVarianceAnalysis} \cred{for SNRs of 25, 35 and 45db}, where the red dashed line depicts the MESMA results. The results show that the values of $\text{PPV}_M$ for \cred{FM-MESMA} were very similar to those obtained by MESMA for the different values of $\sigma_{\calM}^2$ and $K$. Moreover, the PPVs was consistently better for larger $\sigma_{\calM}^2$, what agrees with the conclusions of Theorem~\ref{thm:theorem1}. While a smaller SNR negatively affected the overall accuracy of both algorithms, their relative behavior was not affected. Different values of $K$ also did not have a significant effect on the results obtained by \cred{FM-MESMA} for the selected proportion of large changes.
This suggests that \cred{FM-MESMA} is a computationally efficient and accurate alternative to MESMA in those circumstances.

\begin{figure*}
    \centering
    \includegraphics[width=0.325\linewidth]{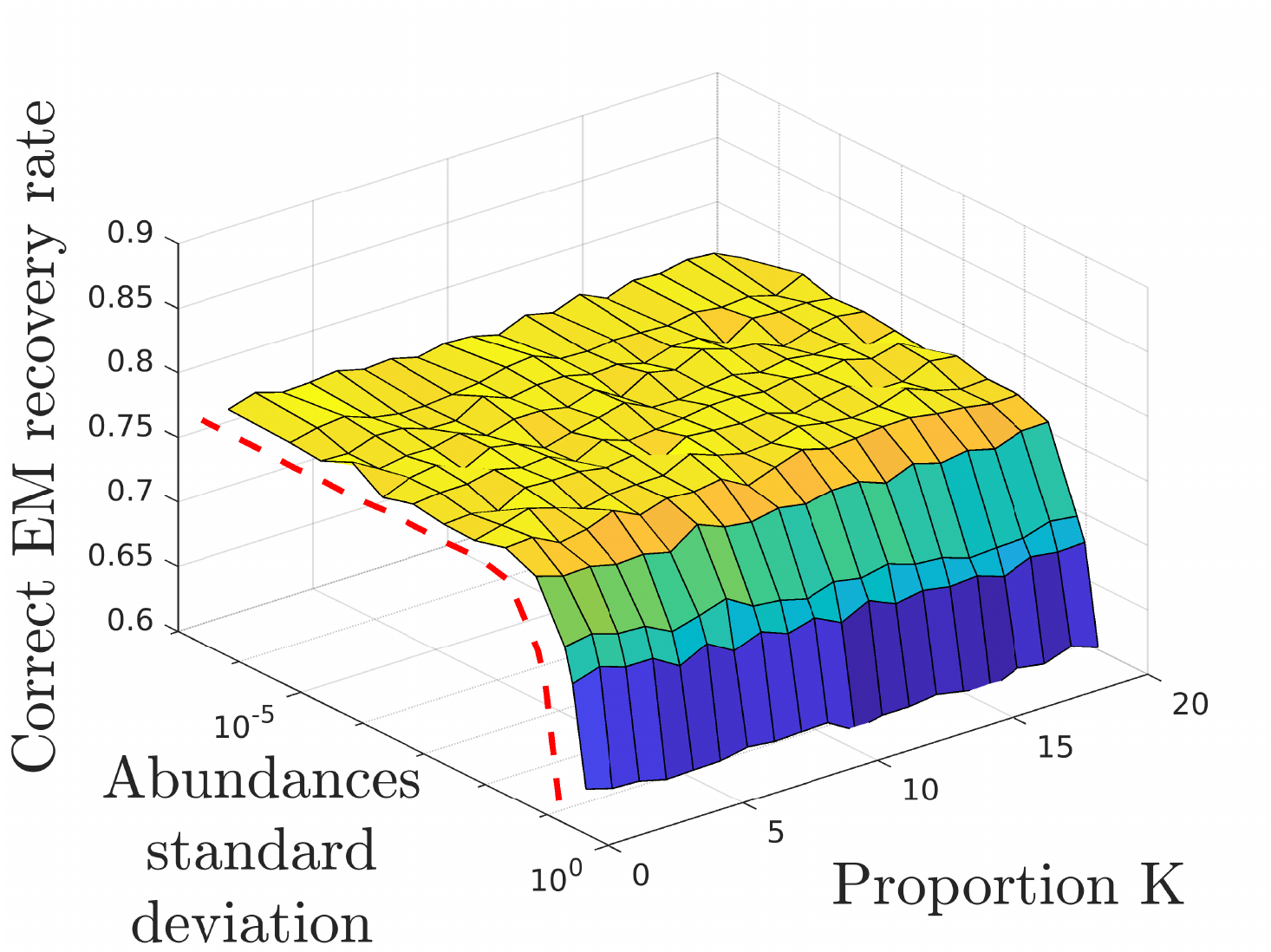}
    \includegraphics[width=0.325\linewidth]{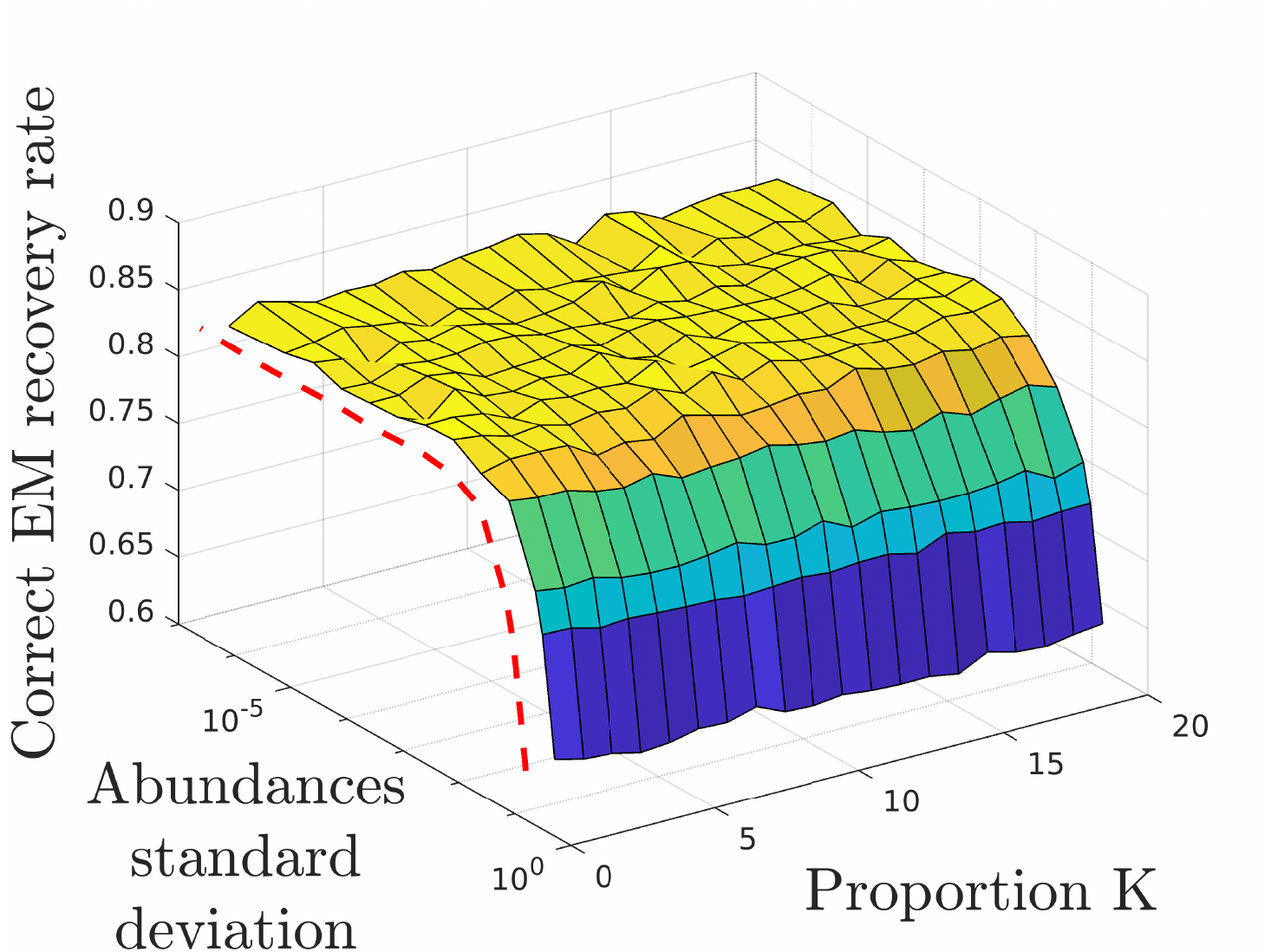}
    \includegraphics[width=0.325\linewidth]{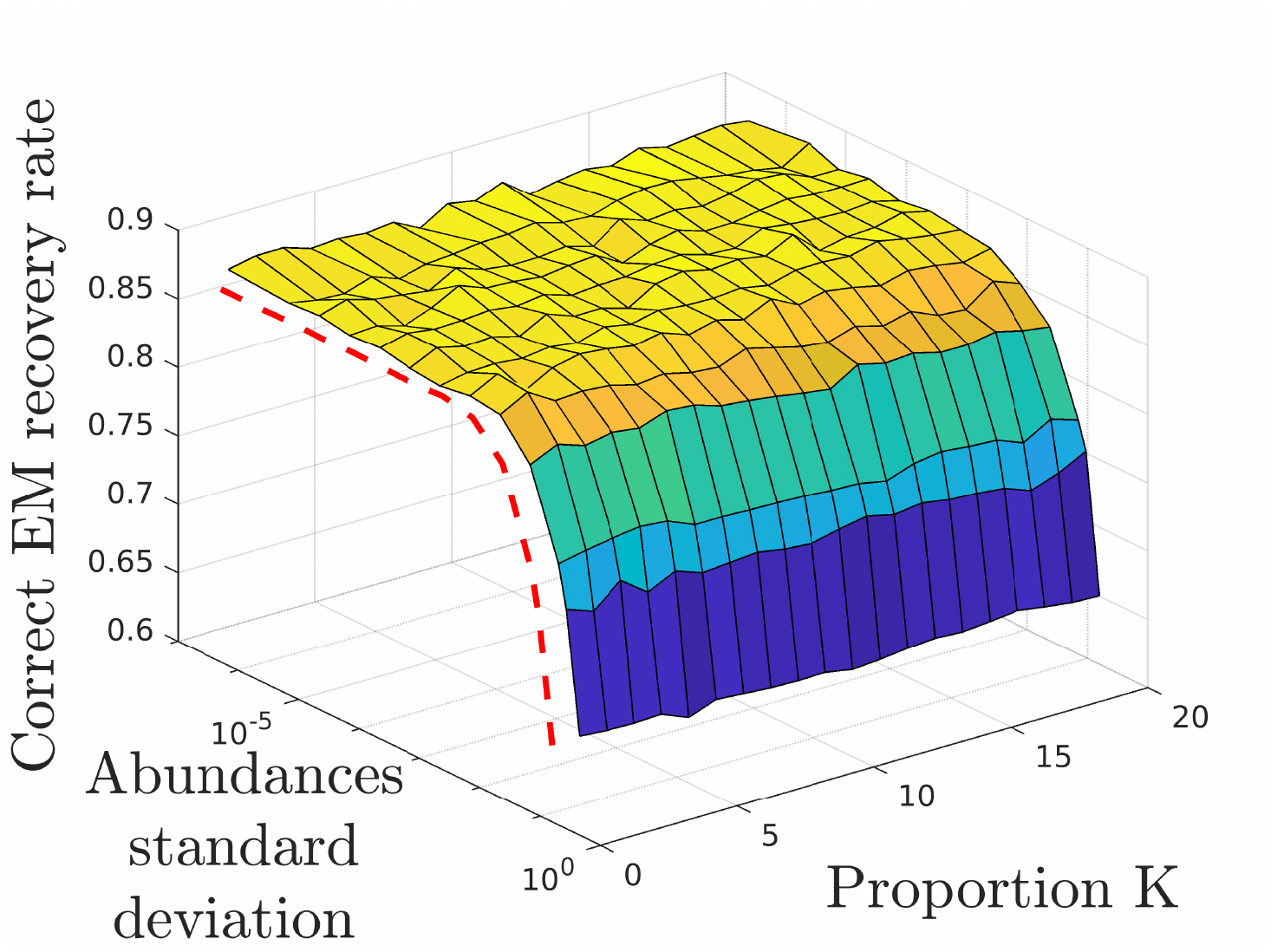}
    \caption{\cred{Abundance variance analysis}: Correct EM recovery rate as a function of abundances standard deviation and \cred{$K$}, for SNRs of 25dB (left), 35dB (middle) and 45dB (right).}
    \label{fig:synthResults_abundancesVarianceAnalysis}
\end{figure*}
\subsubsection{\textbf{Abundance variance analysis}} \label{sec:results_syntheticData_abVar}

We evaluated how often the solution to problem~\eqref{eq:opt_appr_perm_i} was the true EM matrix $\bM_{t,n}$ when compared to that of MESMA, for different amounts of small abundance variances, which are represented by $\bdelta_n$ in~\eqref{eq:multit_model_ib}. This experiment also allowed us to validate numerically one part of Theorem~\ref{thm:theorem1}. We considered $\kappa=0$, $P=4$, $C_{\cred{p}}=3$, and different values of $K\in\{1,\ldots,20\}$. 
The variable abundances were generated by sampling each $\ba_{t,n}$ from Dirichlet distributions with mean $\bmu^{\ba}_{n}$, $n=1,\ldots,N$ and standard deviations ranging from~$10^{-6}$ to~$1$. The abundance mean vectors $\bmu^{\ba}_{n}$ were also sampled from a Dirichlet distribution. The performance were evaluated again using the endmember PPV as defined in~\eqref{eq:synthMetrics_PPV_M}.

The results are shown in Figure~\ref{fig:synthResults_abundancesVarianceAnalysis} for \cred{SNRs of 25, 35 and 45db}, with the red dashed line depicting the MESMA results. It can be seen that for smaller abundance temporal variations (e.g., about $\leq10^{-2}$), \cred{FM-MESMA} was able to obtain an endmember positive predictive value that was high and very similar to that of MESMA, even for larger values of $K$. This again supports the use of \cred{FM-MESMA} as an alternative to MESMA under these circumstances.
For larger temporal variations, a decrease in performance was observed for both MESMA and \cred{FM-MESMA}. This is due to the fact that, when the variance of $\ba_{t,n}$ is large, its density tends to concentrate at the edges of the simplex, what causes many abundance fractions to be close to zero, making the identification of the correct $\bM_{t,n}$ through problems~\eqref{eq:mesma} or~\eqref{eq:opt_appr_perm_i} more difficult.
Although a smaller SNR negatively impacts the overall accuracy of the algorithms (like in the simulations of Section~\ref{sec:results_syntheticData_libVar}), the relative behavior between the algorithms remains approximately the same.

\begin{table}[htb]
\renewcommand{\arraystretch}{1.1}
\footnotesize
\centering
\caption{Average abundance and EM estimation results for the semi-real simulations (values are $\times10^2$).}
\vspace{-0.2cm}
\label{tab:synthetic_semireal_abundance_ex}
\begin{tabular}{l|ccccccccc}
\hline
 & FCLS & OU & MESMA & AAM & Proposed \\\hline
$\text{RMSE}_{\bA}$ & 3.51 & 2.27 & 1.87 & 1.90 & \textbf{1.57} \\
$\text{RMSE}_{\bM}$ & --   & 45.6 & 37.3 & 38.5 & \textbf{36.2} \\
$\text{SAM}_{\bM}$  & --   & 24.3 & \textbf{10.7} & \textbf{10.7} & 11.2 \\
$\text{RMSE}_{\bY}$ & 1.89 & \textbf{0.18} & 0.85 & 0.88 & 0.96 \\
\hline
\end{tabular}
\bigskip
\caption{\cred{Quantitative results for the data in Figure~\ref{fig:composite_vis_semireal} (values are $\times10^2$).}}
\vspace{-0.2cm}
\label{tab:synthetic_semireal_abundance_ex_vis}
\begin{tabular}{l|ccccccccc}
\hline
 & FCLS & OU & MESMA & AAM & Proposed \\\hline
$\text{RMSE}_{\bA}$ & 2.32 & 2.16 & \textbf{1.11} & 1.27 & \textbf{1.11} \\
$\text{RMSE}_{\bM}$ & --   & 35.7 & \textbf{33.8} & 36.4 & 34.0 \\
$\text{SAM}_{\bM}$  & --   & 16.4 & 11.2 & \textbf{10.9} & 11.6 \\
$\text{RMSE}_{\bY}$ & 5.07 & \textbf{0.26} & 0.80 & 0.83 & 0.96 \\
\hline
\end{tabular}
\end{table}

\begin{figure}[h]
    \centering
    \includegraphics[width=1\linewidth]{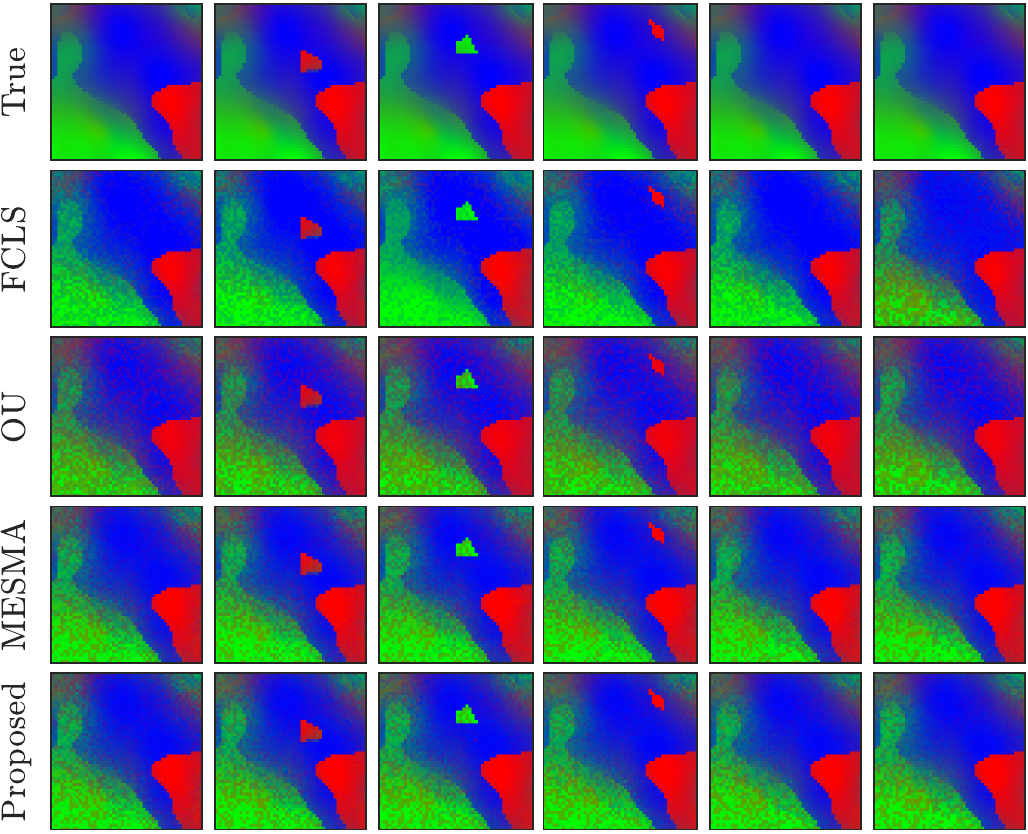}
    \flushright \includegraphics[width=0.79\linewidth]{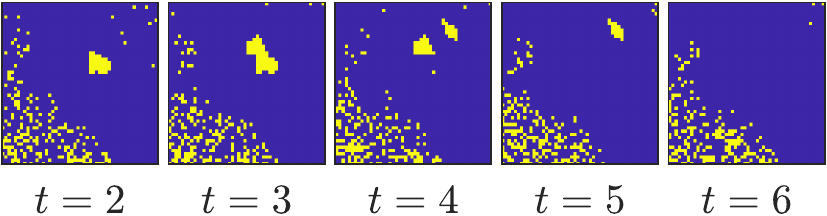}
    \caption{\cred{Composite color image of abundance maps (i.e., abundances for EMs 1, 2 and 3 correspond to red, green and blue colors) in the visual semi-real dataset (top) and changes detected by the proposed method bottom).}}
    \label{fig:composite_vis_semireal}
\end{figure}

\subsection{Abundance estimation performance in semi-real data} 
\label{sec:semi_real_simuls}

We evaluated \cred{FM-MESMA} in terms of abundance and endmember estimation accuracy by comparing it to the FCLS, OU, MESMA, and AAM algorithms using a synthetic data set designed to closely emulate a practical scenario. To this end, first we manually extracted $P=3$ sets of spectrally distinct pure pixels from different materials (tree, road and water) from the Jasper Ridge HI containing six signatures with $L=198$ bands each. Then, we randomly sampled signatures from these sets in order to create one pair of disjoint spectral libraries $\calM_{\cred{p}}^1$, $\calM_{\cred{p}}^2$, i.e., $\calM_{\cred{p}}^1\cap\calM_{\cred{p}}^2=\varnothing$ for each material $\cred{p}=1,\ldots,P$, where both $\calM_{\cred{p}}^1$ and $\calM_{\cred{p}}^2$ have \cred{three signatures each}.
\cred{Afterwards, two sets of temporal abundances were generated. The first one was generated randomly to allow for a statistical evaluation, in which}
$N=1000$ pixels in a sequence of $T=20$ images were generated following the model in~\eqref{eq:multit_model_ia}, where the abundances $\ba_{t,n}$ were sampled from a Dirichlet distribution \cred{and a proportion of $\kappa=0.05$ of the pixels undergoing significant changes. The second set of abundance maps consisted of a single sequence of $T=6$ images containing $N=2500$ pixels whose spatial compositions were adequate for a visual inspection of the results, and large changes were added in the form of random convex polygons. The latter sequence can be seen at the top row of Figure~\ref{fig:composite_vis_semireal}. } 
\cred{For both cases, $\bM_{t,n}$ was sampled uniformly from the library $\calM_{\cred{p}}^1$ and the additive} noise $\be_t$ was selected as white Gaussian with an SNR of 30dB.

The library-based methods, namely, MESMA, AAM and \cred{FM-MESMA}, were then used to unmix the images using the library $\calM_{\cred{p}}^2$, with $K=10$ selected for the proposed method. This generated a mismatch between the signatures in the HI and those used for SU, what is commonly observed in practice. To \cred{reliably} measure the performance of the methods, we ran this simulation over 100 Monte Carlo realizations. The average performance of all algorithms are depicted in Table~\ref{tab:synthetic_semireal_abundance_ex}. It can be observed that the library-based methods provided a considerable improvement in abundance and endmember estimation accuracy when compared to both FCLS and OU. The results of AAM and MESMA were also very similar, with those of AAM being slightly worse. \cred{FM-MESMA} provided an improvement of about $16\%$ in abundance estimation when compared to MESMA. This confirms the benefits of exploring the temporal correlation between the abundances at adjacent time instants. In terms of endmember estimation accuracy, \cred{FM-MESMA} performed similarly to MESMA, with a slightly better $\text{RMSE}_{\bM}$ but slightly worse $\text{SAM}_{\bM}$.
\cred{In terms of $\text{RMSE}_{\bY}$, the library-based methods achieved a similar reconstruction error (with MESMA's being slightly smaller as it solves the EM selection problem exactly), significantly smaller than FCLS but still larger than OU. Since the OU algorithm estimates the EMs from the scene, it has more degrees of freedom and is able to achieve a smaller $\text{RMSE}_{\bY}$. However, it is well-known that smaller reconstruction errors do not necessarily translate into better abundance map estimates.}

\cred{For the second set of abundance maps, the quantitative and visual results are shown in Table~\ref{tab:synthetic_semireal_abundance_ex_vis} and in Figure~\ref{fig:composite_vis_semireal}. Due to space limitations 1) we do not show the AAM results since they were very similar to those of MESMA, and 2) only a composite image is shown, with red, green and blue corresponding to the abundances of the first, second, and third EMs, respectively. The quantitative metrics show that the algorithms behaved very similarly to the previous case, with the minor differences that MESMA and AAM now performed slightly better in terms of $\text{RMSE}_{\bA}$ and $\text{RMSE}_{\bM}$, respectively, when compared to the remaining methods. Visually, it can be seen that the results of MESMA and of FM-MESMA are similar and approach the ground truth more closely when compared to FCLS and OU, whose results appear less accurate and more skewed towards blue for the FCLS, and towards red for OU when compared to the ground truth. The sudden changes in the ground truth sequence were also well captured by FM-MESMA. However, various pixels in the lower-left segment were also incorrectly marked as changes, indicating that the selection of the proportion $K$ was slightly conservative.}

\begin{figure}
    \centering
    \includegraphics[width=\linewidth]{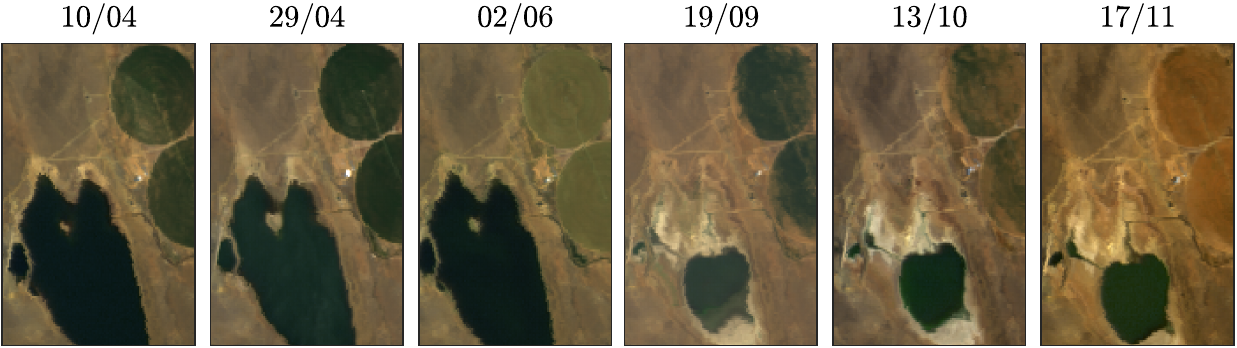}
    \caption{Lake Tahoe HI sequence (heading text means day/month).}
    \label{fig:lake_tahoe_HIs}
\end{figure}

\begin{figure}
    \centering
    \includegraphics[width=\linewidth]{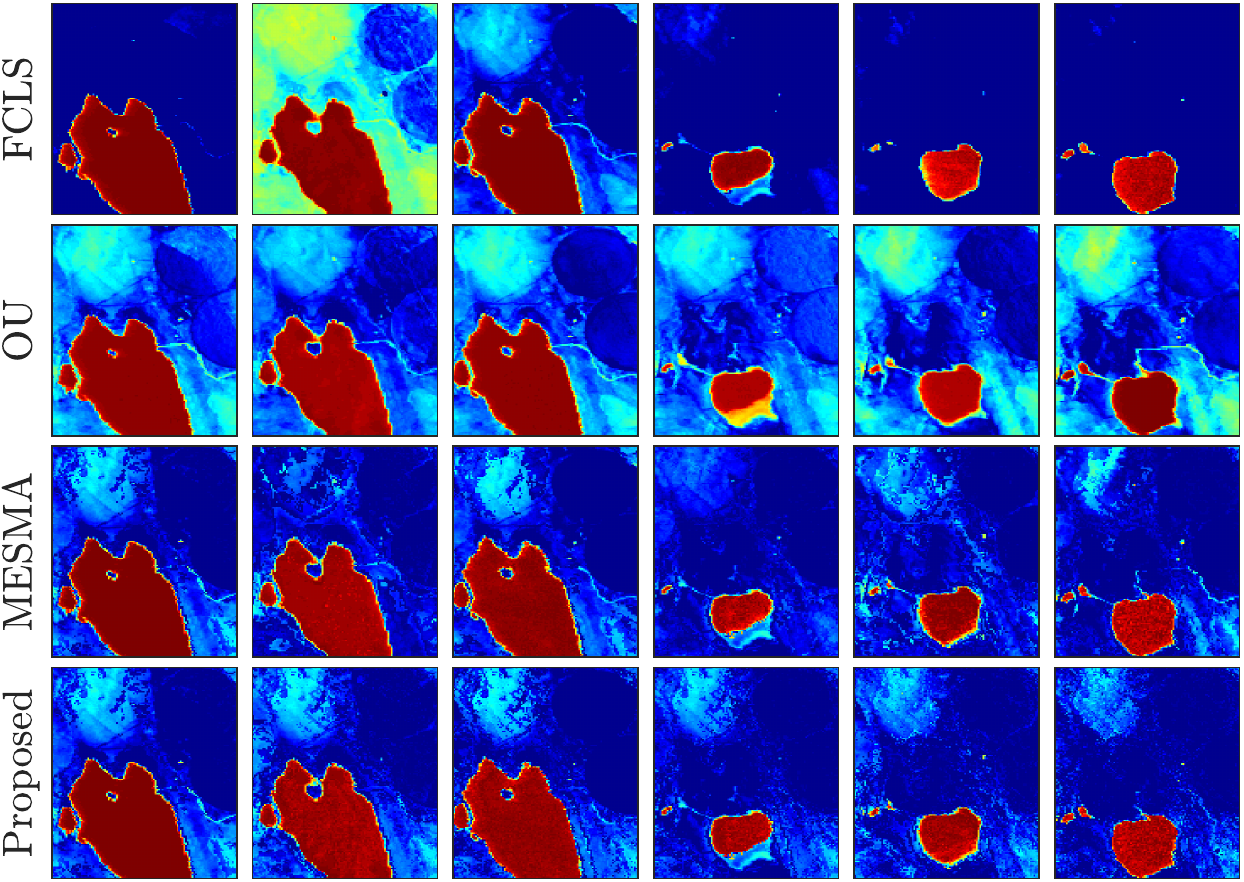}
    \caption{Multitemporal \cred{abundances} of the Lake Tahoe HI for the water EM.}
    \label{fig:abundances_tahoe_water}
\end{figure}
\begin{figure}
    \centering
    \includegraphics[width=\linewidth]{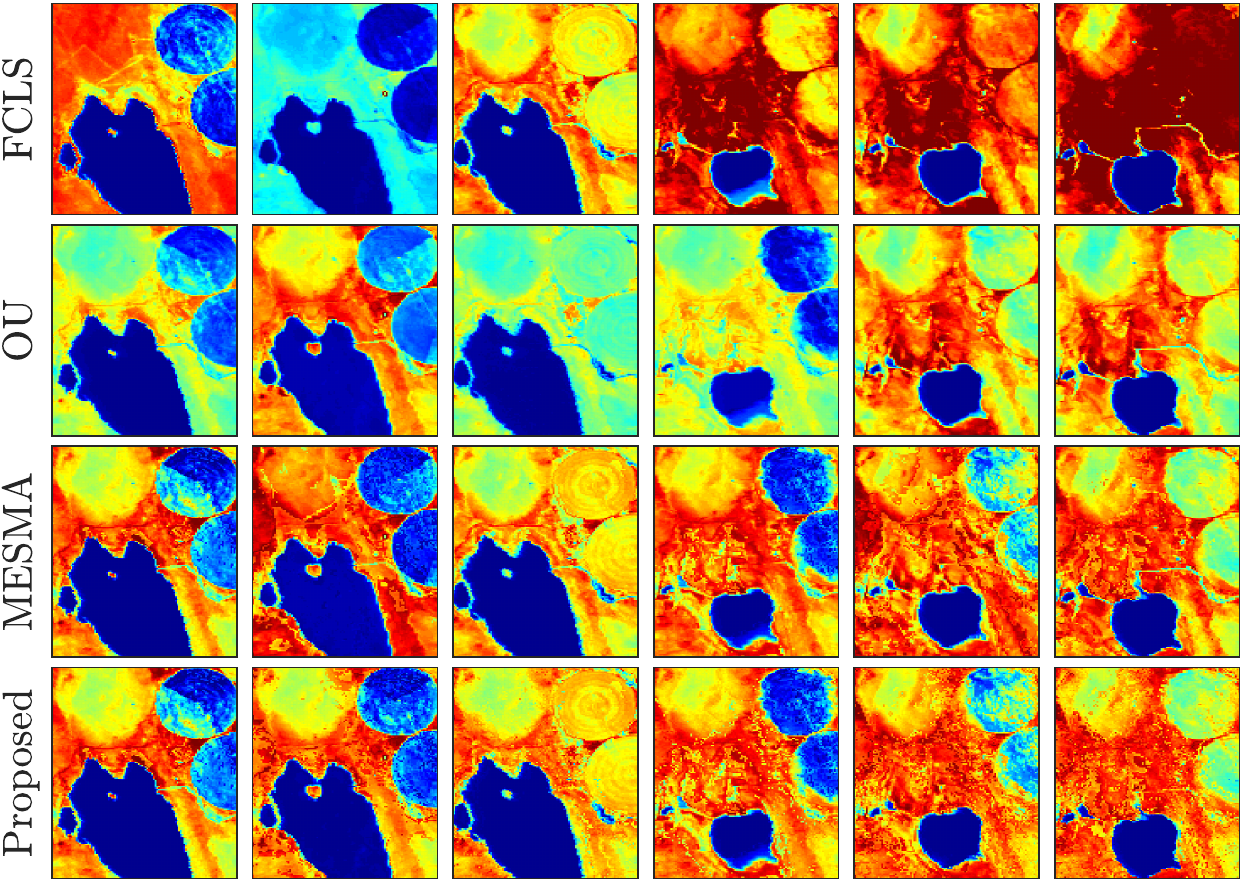}
    \caption{Multitemporal \cred{abundances} of the Lake Tahoe HI for the soil EM.}
    \label{fig:abundances_tahoe_soil}
\end{figure}
\begin{figure}
    \centering
    \includegraphics[width=\linewidth]{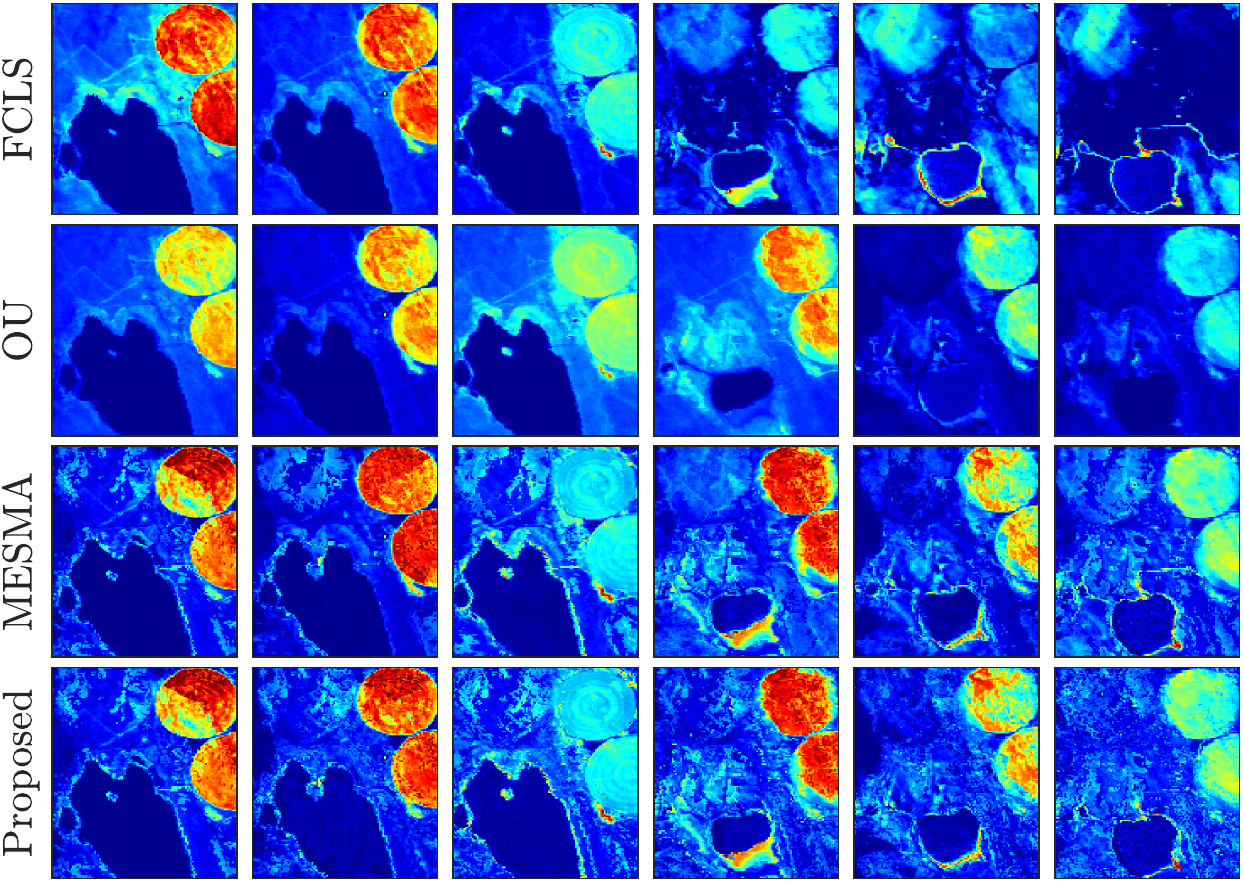}
    \caption{Multitemporal \cred{abundances} of the Lake Tahoe HI for the vegetation EM.}
    \label{fig:abundances_tahoe_vegetation}
\end{figure}
\begin{figure}
    \centering
    \includegraphics[width=\linewidth]{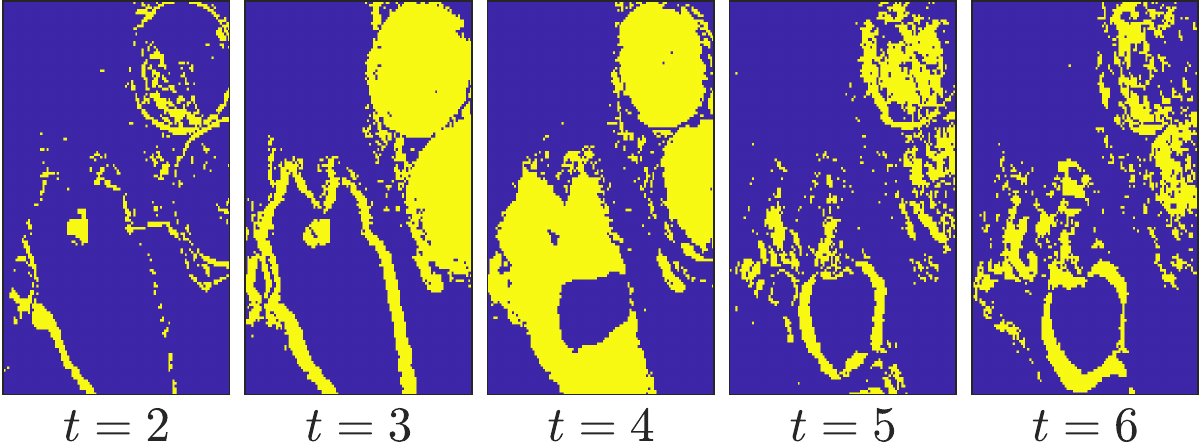}
    \caption{Detected changes for the Lake Tahoe HIs.}
    \label{fig:tahoe_changes}
\end{figure}

\subsection{Real Data}

For the simulations with real data, we considered the Lake Tahoe data set, which was originally presented in~\cite{Thouvenin_IEEE_TIP_2016}. This data set consists of $T=6$ images acquired by the AVIRIS instrument, each with $224$ bands and $N=16500$ pixels. Water absorption and low SNR bands were removed, resulting in $L=173$ bands. A false color representation is shown in~\ref{fig:lake_tahoe_HIs}, where three predominant materials (soil, water and vegetation) can be identified. The images were ordered according to the day/month in which they were acquired in order to make the seasonal changes more clear.
The EM libraries for MESMA and for \cred{FM-MESMA} were constructed as follows. First, reference EMs for each material were extracted by applying the VCA to the concatenation of the pixels of all six images. Then, a preliminary library was constructed by extracting pixels from each image that had small spectral angle to the reference endmembers. Finally, the libraries were obtained by removing/pruning the most redundant signatures (as measured according to their mutual Euclidean distance) so that the final libraries contained $C_{\cred{p}}=6$ signatures for each EM. We set $K=10$ for \cred{FM-MESMA}, and the parameters for the OU algorithm were the same as those used in~\cite{Thouvenin_IEEE_TIP_2016}.

The abundance maps estimated by the algorithms are shown in Figures~\ref{fig:abundances_tahoe_water},~\ref{fig:abundances_tahoe_soil} and~\ref{fig:abundances_tahoe_vegetation}. Due to lack of space, we do not show the AAM results since they were very similar to those of MESMA. It can be seen that the FCLS performed poorly for some of the frames (e.g., the second and the fourth ones), with significant confusion between different materials. 
Although OU \cred{showed} more consistent results, they \cred{were} not as good as those by MESMA \cred{or FM-MESMA}. Specifically, the OU abundances \cred{did} not show a separation between vegetation and soil as clear as observed in Figure~\ref{fig:lake_tahoe_HIs}. Moreover, significant water abundances \cred{were predicted} outside of the lake. The results of MESMA and \cred{FM-MESMA were} very similar, and closely \cred{agreed} with the distribution of the endmembers observed in Figure~\ref{fig:lake_tahoe_HIs}. However, some slight differences between the methods can be noticed, notably the abundances by the proposed algorithm being smoother in time as in the case of the water endmember.
The changes detected by \cred{FM-MESMA}, shown in Figure~\ref{fig:tahoe_changes}, clearly \cred{distinguish} the seasonal variations at the borders of the lake and in the crop circles due to different water levels and in the growth states, respectively.

The reconstruction errors $\text{RMSE}_{\bY}$ and execution times of the algorithms are provided in Table~\ref{tab:times_real_data}. 
\cred{The $\text{RMSE}_{\bY}$ results behaved very similarly to the semi-real case discussed in Section~\ref{sec:semi_real_simuls}, with OU achieving the smallest and FCLS the largest reconstruction errors, and the library-based methods performing similarly to each other.}
The execution time of \cred{FM-MESMA} was about half that of MESMA even though there were significant changes between some pairs of frames, what agrees with the results in Table~\ref{tab:comparative_synthetic_execution_times}. Moreover, this scene contains a small number of materials and is thus not a good representative of the relative performance between the algorithms in other scenarios. For image sequences with larger values of $P$, the computational complexity gains will be much more significant.

\begin{table}[!t]
\renewcommand{\arraystretch}{1.2}
\footnotesize
\centering
\caption{Execution times and $\text{RMSE}_{\bY}$ for the Lake Tahoe HI.}
\vspace{-0.2cm}
\label{tab:times_real_data}
\begin{tabular}{c|ccccccccc}
\hline
 & FCLS & OU & MESMA & AAM & Proposed \\\hline
Time [s] & 12.3 & 85.5 & 112.0 & 228.8 & 62.8 \\
$\text{RMSE}_{\bY}$ & 0.412 & 0.013 & 0.064 & 0.068 & 0.072 \\
\hline
\end{tabular}
\end{table}

\section{Conclusions} \label{sec:conclusions}

In this paper, we presented a new computationally efficient multitemporal unmixing algorithm \cred{(FM-MESMA)} based on multiple endmembers spectral mixture analysis. The proposed strategy exploits the high temporal correlation of the abundance maps in order to improve both the accuracy and the computational complexity of the algorithm. Specifically, it approximates the solution to the multitemporal unmixing problem by separating it into two sub-problems, namely, endmember selection and abundance estimation, which are much easier to solve individually. A strategy was also proposed to detect abrupt abundance changes by analysing residuals of the endmember selection problem. Theoretical results demonstrated how \cred{FM-MESMA} compares to MESMA in terms of quality and effectiveness in detecting abrupt abundance changes. Besides, these results also provide valuable insight into the conditions under which the approximate algorithm succeeds. Simulation results showed that the proposed method gives results with quality similar to, or better than, both MESMA and parametric models at a reduced computational complexity.

\bibliographystyle{IEEEtran}
\bibliography{references_mtemp,references_revpaper}

\end{document}